\documentclass[10pt]{article}

\usepackage[top=1in, bottom=1.5in, left=1in, right=1in]{geometry}
\usepackage{authblk} 
\date{}
\usepackage{amsmath,amsthm,amssymb,amsfonts}
\usepackage[braket, qm]{qcircuit}
\usepackage[backref]{hyperref}
\usepackage{bm}
\newtheorem{thm}{Theorem}
\newtheorem{lem}{Lemma}
\newtheorem{remark}{Remark}
\newtheorem{cor}{Corollary}
\usepackage{algorithm}
\usepackage{appendix}
\usepackage{algorithmicx}
\usepackage{algpseudocode}
\usepackage{graphicx}
\usepackage{color}

\definecolor{lzcol}{rgb}{1, 0, 0}

\title{An Efficient Explicit Implementation of a Quantum Algorithm with Quantum Advantage for Nonlinear Scalar Conservation Laws}

\author[1,2]{Kezhen Wang}
\author[3]{Junpeng Hu}
\author[1,2,3,*]{Lei Zhang}
\affil[1]{School of Mathematical Sciences,   Shanghai Jiao Tong University, Shanghai 200240, China.}
\affil[2]{Institute of Natural Sciences, Shanghai Jiao Tong University, Shanghai 200240, China.}
\affil[3]{UnitaryLab Quantum Technology Co., Ltd., Shanghai, P.R. China, 2026}

\begin{document}

\maketitle
\begin{abstract}
Quantum algorithms for nonlinear partial differential equations remain challenging because nonlinear dynamics are not directly amenable to unitary quantum simulation. Building on the level-set formulation of \cite{Jin2022QuantumAF}, we construct a quantum algorithm and provide an explicit gate-level implementation for solving scalar conservation laws. The nonlinear equation is first lifted to a linear Liouville equation, discretized by finite differences, and then embedded into a unitary evolution through Schr\"odingerisation. We further develop quantum procedures for estimating relevant observables from the evolved state. Error bounds and gate-complexity estimates are established for the complete algorithm. The resulting complexity comparison demonstrates quantum advantage for observable estimation in sufficiently high spatial dimensions, under standard assumptions on state preparation and oracle access. Finally, numerical experiments validate the accuracy, multidimensional applicability, and predicted scaling of the proposed method.
\end{abstract}


\section{Introduction}
\label{sec:introduction}

\par Hyperbolic conservation laws provide a fundamental mathematical framework for transport and wave propagation in physical systems. Scalar conservation laws are the simplest nonlinear representatives of this class and arise in compressible flow and shock dynamics through the Burgers equation \cite{courant1999shock, leveque2002hyberbolic}, traffic-flow modeling through the Lighthill--Whitham--Richards model \cite{Lighthill1955traffic}, and two-phase flow in porous media through the Buckley--Leverett equation \cite{Buckley1942sand}. They also serve as basic models for more complicated systems in fluid dynamics, plasma physics, and other transport-dominated applications.

\par We consider the Cauchy problem
\begin{equation}
\begin{cases}
\dfrac{\partial u}{\partial t} + \nabla \cdot \mathbf{F}(u) = 0, & \mathbf{x} \in \mathbb{R}^d, \; t > 0, \\
u(\mathbf{x},0) = u_0(\mathbf{x}),
\end{cases}
\end{equation}
where $u(\mathbf{x}, t):\mathbb{R}^d\times \mathbb{R}^+\rightarrow \mathbb{R}$ represents the conserved scalar quantity and $\mathbf{F}(u)=\left(F_1(u),F_2(u),\dots,F_d(u)\right)$ denotes the flux function.

\par Even for smooth initial data, nonlinear characteristics may intersect in finite time and produce shocks \cite{Lax1987HyperbolicSO}. The level-set method provides a natural framework for representing these branches in an augmented phase space. It lifts the nonlinear conservation law to a linear Liouville equation \cite{Levelset}, which forms the starting point of the quantum algorithm developed in this work. On the computational side, standard finite-volume, discontinuous Galerkin, and WENO methods are highly effective in low spatial dimensions \cite{shu1988weno, cockburn1998rungekutta}, but a direct grid discretization requires a number of degrees of freedom that grows exponentially with the dimension. This scaling motivates alternative approaches for high-dimensional transport problems.

\par Quantum algorithms provide a possible route to manipulating such large discretized systems in an amplitude-encoded form. Unitary dynamics can be implemented through Hamiltonian simulation \cite{Lloyd1996UniversalQS,Berry_2015Taylor,Low_2017_QSP}, whereas the non-unitary evolution produced by many dissipative or upwind discretizations requires an additional embedding. Schr\"{o}dingerisation \cite{SchrodingerJin2022, Schrodinger2} converts a general linear differential system into a higher-dimensional Schr\"{o}dinger-type equation by introducing an auxiliary variable. The resulting Hermitian dynamics can then be simulated using product formulas or more advanced Hamiltonian-simulation techniques. In this work we adopt the first-order Lie--Trotter product formula \cite{trotter1}, with its error controlled through commutator estimates \cite{trottererror}, because it preserves the sparse tensor-product structure of the finite-difference operators and admits an explicit gate-level realization.

\par The principal obstacle is the nonlinearity of the original conservation law. Existing quantum approaches to nonlinear differential equations include Carleman-type linearization \cite{Lloyd2020QuantumAF, Leyton2008AQA} and phase-space formulations related to Koopman--von Neumann dynamics \cite{Joseph2020KoopmanvonNA}. Here we use the level-set lifting of \cite{Levelset, Jin2022QuantumAF}. By introducing an additional scalar variable $p$, the nonlinear conservation law is represented by a linear Liouville equation whose transport coefficients are the sampled flux derivatives. After an upwind finite-difference discretization, this equation becomes a block-diagonal linear ODE in the $p$-index and is therefore directly compatible with Schr\"{o}dingerisation.

\par We then construct the corresponding quantum evolution at the circuit level. The spatial forward and backward difference operators are decomposed into structured shift terms and implemented using multi-controlled rotations and CNOT gates following \cite{Sato2024HamiltonianSF}. The dependence on the auxiliary $p$-register is incorporated through controlled coefficient-loading procedures, including binary-controlled operations and LCU or block-encoding constructions when appropriate.

\par Overall, this work makes three main contributions. First, we provide an end-to-end quantum implementation for solving nonlinear scalar conservation laws, from the level-set linearization and finite-difference discretization to the Schr\"odingerised Hamiltonian evolution and its Trotterized gate-level circuits. Second, we develop quantum procedures for encoding the \(p\)-dependent flux coefficients, recovering the physical solution, and estimating relevant observables from the evolved state. Third, we establish error and gate-complexity estimates for the complete algorithm, which demonstrate a quantum advantage for observable estimation in sufficiently high spatial dimensions, under standard assumptions on state preparation and oracle access.

\par \emph{Organization}. Section~\ref{sec:levlsetlifting} introduces the level-set lifting and its finite-difference discretization. Section~\ref{sec:quantumSCL} presents Schr\"{o}dingerisation, the quantum difference operators, the complete time-evolution circuit, and observable recovery. Section~\ref{sec:error} gives the error and complexity analysis. Section~\ref{sec:numerics} reports the numerical experiments for representative one- and two-dimensional scalar conservation laws, which validate the accuracy and predicted scaling of the proposed method.

\section{Level Set Lifting and Discretization}
\label{sec:levlsetlifting}
\par This section establishes the analytical and numerical foundation for the quantum algorithm developed in Section~\ref{sec:quantumSCL}. We first review the level set method, which lifts the nonlinear scalar conservation law to a linear Liouville equation in a higher-dimensional phase space. This lifting is the key step that renders the problem compatible with quantum simulation techniques. We then introduce the distribution-function formulation, which provides a more tractable representation of the lifted dynamics, and describe its discretization by an upwind finite difference scheme. The resulting semi-discrete system is a linear ODE, which serves as the direct input for the Schr"{o}dingerisation procedure and the quantum circuit construction in the next section.

\subsection{Level set method}
\label{subsec:levelset}
We reformulate the nonlinear conservation law by introducing a level set function $\bm{\phi}$, which represents the solution implicitly. This yields a linear Liouville equation for $\bm{\phi}$, allowing the nonlinear PDE to be embedded into a linear framework that is suitable for quantum simulation.

Consider the scalar conservation law
\begin{equation}\label{conservation law}
\begin{cases}
\partial_t u + \nabla \cdot \mathbf{F}(u) = 0, & (x,t) \in \mathbb{R}^d \times (0,\infty), \\
u(x,0) = u_0(x), & x \in \mathbb{R}^d.
\end{cases}
\end{equation}

We introduce the level set function $\bm{\phi}(t,x,p):\mathbb{R}^+\times \mathbb{R}^d\times \mathbb{R} \to \mathbb{R}$ in $(d+2)$ dimensions, implicitly defined by $\bm{\phi}(t,x,p) = 0 \quad \text{with } p = u(t,x)$. Then $\bm{\phi}$ satisfies the Liouville equation \cite{Levelset}
\begin{equation}\label{CL}
\partial_t \bm{\phi} + (\nabla_p \mathbf{F}(p)) \cdot \nabla_x\bm{\phi} = 0,
\end{equation}
subject to the initial condition $\bm{\phi}(0,x,p) = p - u_0(x)$.

\par In general, solutions of the conservation law may become multi-valued after shock formation. The corresponding branches can be represented as $u_\nu(t,x) = p_\nu(t,x), \quad \nu = 1,\dots,\mathcal{J}$, where $\{p_\nu(t,x)\}_{\nu=1}^{\mathcal{J}}$ are the roots of $\bm{\phi}(t,x,p)=0$. Through this lifting, the original $(d+1)$-dimensional nonlinear PDE is transformed into a $(d+2)$-dimensional linear Liouville equation.

\par\textbf{Distribution function formulation.}
An equivalent and computationally convenient formulation is obtained by introducing the scalar distribution function $\psi$ defined by
$ \psi(t,x,p) = \delta(\bm{\phi}(t,x,p))$, where $\delta(\cdot)$ denotes the Dirac delta function. Then $\psi$ satisfies the linear transport equation
\begin{equation}\label{psi}
\partial_t \psi + \nabla_p \mathbf{F}(p) \cdot \nabla_x \psi = 0,
\end{equation}
with initial condition 
\begin{equation}\label{psi0}
\psi(0,x,p) = \delta(p - u_0(x)).
\end{equation}

When the solution becomes multi-valued, by \cite[Lemma 9]{Jin2022QuantumAF}, $\psi$ admits the representation 
\begin{equation}\label{eq:multivalued}
    \psi(t, x, p) = \sum_{\nu=1}^{\mathcal{J}} \frac{1}{|\partial_p \bm{\phi}(t, x, p_\nu)|} \, \delta(p - p_{\nu}(t, x)),
\end{equation}
where $\{p_\nu(t, x)\}_{\nu=1}^{\mathcal{J}}$ are the $\mathcal{J}$ roots of $\bm{\phi}(t, x, p) = 0$.

\subsection{Discretization}\label{subsec:Discretization}

\par We introduce a uniform discretization of the $(d+1)$-dimensional phase space $(\mathbf{x},p)\in \Omega_x\times \Omega_p \subset \mathbb{R}^{d+1}$:
\begin{itemize}
    \item $\bm{x}$-space: $N_x$ grid points per dimension, spacing $h = 1/N_x$, grid points $\bm{x}_{\bm{j}} = h \bm{j}$ with $\bm{j} = (j_1, \dots, j_d)$, $j_\alpha = 0, \dots, N_x-1$.
    \item $p$-space: $N_p$ grid points, spacing $h_p = (p_{\max} - p_{\min})/N_p$, grid points $p_l = p_{\min} + l h_p$ with $l = 0, \dots, N_p-1$.
\end{itemize}
Without loss of generality, we take $\Omega_x=[0,1]^d$. The continuous distribution function is approximated by its values on the phase-space grid:
\begin{equation}\label{eq:grid_approx}
    \psi(t, \bm{x}_{\bm{j}}, p_l) \approx \psi_{\bm{j}, l}(t).
\end{equation}

\par For the Liouville equation \eqref{psi}, we discretize the spatial derivatives using a first-order upwind finite difference scheme. The derivative $\partial_{x_\alpha}\psi$ is approximated by
\begin{equation}\label{eq:spatial_derivative}
    \partial_{x_\alpha} \psi \approx
    \begin{cases}
        \dfrac{\psi_{\bm{j}+\bm{e}_\alpha,l} - \psi_{\bm{j},l}}{h}, & F'_\alpha(p_{l}) < 0, \\[10pt]
        \dfrac{\psi_{\bm{j},l} - \psi_{\bm{j}-\bm{e}_\alpha,l}}{h}, & F'_\alpha(p_{l}) > 0,
    \end{cases}
\end{equation}
where $\bm{e}_\alpha$ denotes the unit vector in the $\alpha$-th coordinate direction and $F'_\alpha(p) = \partial_{p} F_{\alpha}(p)$, with $\alpha = 1, \cdots, d$. Decomposing the flux derivative into its positive and negative parts,
\begin{equation}\label{eq:flux_split}
    F'_\alpha(p) = [F'_\alpha(p)]_+ + [F'_\alpha(p)]_-, \qquad
    [a]_+ = \max(a, 0), \quad [a]_- = \min(a, 0),
\end{equation}
we obtain the following semi-discrete linear ODE system:
\begin{equation}\label{eq:semi_discrete}
    \frac{d}{dt} \psi_{\bm{j},l}(t) = -\sum_{\alpha=1}^{d} \left( [F'_\alpha(p_{l})]_- (D_{+}^{(\alpha)} \bm{\psi})_{\bm{j},l} + [F'_\alpha(p_{l})]_+ (D_{-}^{(\alpha)} \bm{\psi})_{\bm{j},l} \right),
\end{equation}
where $D_{+}^{(\alpha)}$ and $D_{-}^{(\alpha)}$ denote the forward and backward difference operators along the $\alpha$-th spatial dimension.

\par Collecting all $\psi_{\bm{j},l}$ into a vector $\bm{\Psi}(t)\in \mathbb{R}^{N_p N_x^{d}}$, the semi-discrete system \eqref{eq:semi_discrete} can be written compactly as
\begin{equation}\label{eq:PsiOde}
    \frac{d\bm{\Psi}}{dt}=\bm{A}\bm{\Psi}(t),\quad \bm{\Psi}(0)=\bm{\Psi}_0,
\end{equation}
where $\bm{\Psi}_0$ represents the discretized initial condition and $\bm{A}\in \mathbb{R}^{N_pN_x^{d}\times N_pN_x^{d}}$ is block-diagonal with respect to the $p$-index:
\begin{equation}\label{eq:A_structure}
    \bm{A} = \bigoplus_{l=0}^{N_p-1} \bm{A}^{(l)}, \qquad
    \bm{A}^{(l)} = \sum_{\alpha=1}^{d} \big( [F'_\alpha(p_l)]_- D_+^{(\alpha)} + [F'_\alpha(p_l)]_+ D_-^{(\alpha)} \big).
\end{equation}
Each block $\bm{A}^{(l)}$ acts only on the spatial degrees of freedom associated with a fixed level-set grid point $p_l$. Consequently, different $p$-modes evolve independently, with the corresponding spatial transport operator determined by the local flux derivatives $F'_\alpha(p_l)$.

\par\textbf{Initial condition: $\delta$-function regularization.}
Since the initial condition \eqref{psi0} contains a Dirac delta distribution, it must be replaced by a smooth approximation for numerical implementation. Let $\delta_\omega : \mathbb{R} \to \mathbb{R}$ be a compactly supported mollifier satisfying
\begin{equation}\label{eq:delta_properties}
    \delta_\omega(x) = 0 \;\; \text{for } |x| > \omega, \qquad \int_{-\omega}^{\omega} \delta_\omega(x) \, dx = 1,
\end{equation}
where $\omega > 0$ denotes the smoothing parameter. Common choices include the tent function and the raised cosine:
\begin{equation}\label{eq:delta_examples}
    \delta_\omega(x) = \frac{1}{\omega}\left(1 - \frac{|x|}{\omega}\right)\mathbf{1}_{[-\omega,\omega]}(x), \qquad
    \delta_\omega(x) = \frac{1}{2\omega}\left(1 + \cos\left(\frac{\pi x}{\omega}\right)\right)\mathbf{1}_{[-\omega,\omega]}(x).
\end{equation}
The regularized initial condition is then given by
\begin{equation}\label{eq:regularized_ic}
    \psi_{\bm{j},l}(0) = \delta_\omega\big(p_{l} - u_0(x_{\bm{j}})\big).
\end{equation}
To ensure accurate quadrature in the $p$-variable, the support of $\delta_\omega$ should cover at least $m \geq 2$ grid cells in the $p$-direction:
\begin{equation}\label{eq:omega_h}
    \omega = m h_p, \qquad m \geq 2,
\end{equation}
which gives $N_p = \mathcal{O}((p_{\max} - p_{\min}) / \omega)$.

\section{Quantum Algorithm of Scalar Conservation Laws}
\label{sec:quantumSCL}

\par In this section, we present a quantum algorithm for scalar conservation laws based on the level set--Liouville formulation developed in Section~\ref{sec:levlsetlifting}. Starting from the linear ODE system \eqref{eq:PsiOde}, we apply Schr\"{o}dingerisation to embed the non-unitary dynamics into a higher-dimensional Hamiltonian system, and then construct explicit gate-level quantum circuits for the resulting time evolution.


\subsection{Quantum representation of finite difference operators}
\label{subsec:quantumdiff}

\par We first describe how the finite difference operators arising from the discretization in Section~\ref{subsec:Discretization} can be encoded into quantum circuits, following \cite{Sato2024HamiltonianSF}. Consider a one-dimensional domain $\Omega := (0,L)$ discretized into $N=2^n$ uniformly spaced grid points with spacing $h = L/(N+1)$, and let $\mathbf{u} = (u_0,\dots,u_{N-1})^\top$ denote the discrete solution vector. The standard finite difference operators are
\begin{equation*}
\begin{aligned}
    (\mathbf{D}^{+} \mathbf{u})_j &= \frac{u_{j+1}-u_j}{h}, \quad
    (\mathbf{D}^{-} \mathbf{u})_j = \frac{u_j-u_{j-1}}{h},\\
    (\mathbf{D}^{\pm} \mathbf{u})_j &= \frac{u_{j+1}-u_{j-1}}{2h}, \quad
    (\mathbf{D}^{\Delta} \mathbf{u})_j = \frac{u_{j+1}-2u_j+u_{j-1}}{h^2},
\end{aligned}
\end{equation*}
where the boundary values are specified by the imposed boundary conditions; for example, under periodic boundary conditions, $u_N = u_0$ and $u_{-1} = u_{N-1}$.

\par These operators admit tensor-product representations in terms of the elementary matrices $\sigma_{01} = |0\rangle\langle1|$, $\sigma_{10} = |1\rangle\langle0|$, and $I$. Defining the qubit shift operators
\begin{equation*}
    S^- = \sum_{j=1}^{n} s_j^-, \quad
    S^+ = \sum_{j=1}^{n} s_j^+,
\end{equation*}
with components
\begin{equation*}
    s_j^- = I^{\otimes(n-j)} \otimes \sigma_{01} \otimes \sigma_{10}^{\otimes(j-1)}, \quad
    s_j^+ = I^{\otimes(n-j)} \otimes \sigma_{10} \otimes \sigma_{01}^{\otimes(j-1)},
\end{equation*}
the finite difference operators can be expressed as combinations of $S^-$ and $S^+$. For example, the forward difference operator with periodic boundary conditions is given by
\begin{equation*}
    D_P^+ = \frac{1}{h}\left(S^- - I^{\otimes n} + \sigma_{10}^{\otimes n}\right).
\end{equation*}

\par Following \cite{Sato2024HamiltonianSF}, these shift operators can be implemented through the Hamiltonian
\begin{equation*}
    \mathcal{H} = \gamma \sum_{j=1}^n \left(e^{i\lambda}s_j^- + e^{-i\lambda}s_j^+\right), \qquad \gamma,\lambda \in \mathbb{R},
\end{equation*}
whose time evolution $\exp(-i\mathcal{H}\tau)$ is approximated by the unitary circuit $V(\gamma\tau,\lambda)$ constructed from the multi-qubit gates
\begin{equation*}
    U_j(\lambda) = \left(\prod_{m=1}^{j-1} \mathrm{CNOT}_m^j \right) P_j(\lambda) H_j.
\end{equation*}
Here, $P_j(\lambda)$ is the phase gate acting on the $j$-th qubit, and $\mathrm{CNOT}_m^j$ denotes the $\mathrm{CNOT}$ gate acting on the $m$-th qubit controlled by the $j$-th qubit. The circuit $V(\gamma\tau, \lambda)$ is given by
\begin{equation}
    V(\gamma\tau, \lambda)=U_j(\lambda)CRZ^{1,\dots j-1}_{j}(-2\gamma\tau)U_j(\lambda )^\dagger
\end{equation}
where $CRZ^{1,\dots j-1}_{j}$ denotes the $RZ$ gate acting on the $j$-th qubit controlled by the $1,\dots,j-1$-th qubits. The following lemma bounds the approximation error.

\begin{lem}[\cite{Sato2024HamiltonianSF}]
Let $\mathcal{H}=\gamma\sum_{j=1}^n\left(e^{i\lambda}s_j^{-}+e^{-i\lambda}s_j^{+}\right)$. The time evolution operator $\exp(-i\mathcal{H}\tau)$ satisfies
\begin{equation*}
\left\| \exp(-i\mathcal{H}\tau) - V(\gamma\tau, \lambda) \right\|
\leq \frac{\gamma^2 \tau^2 (n-1)}{2}.
\end{equation*}
\end{lem}

\subsection{Quantum circuit implementation}
\label{subsec:circuits}

\par We now present the full quantum circuit implementation, organized into three components: quantum state encoding, Schr\"{o}dingerisation, and explicit circuit construction for the time-evolution operator.

\subsubsection{Quantum encoding}
\par Following the discretization in Section~\ref{subsec:Discretization}, the dynamics are governed by the linear ODE system \eqref{eq:PsiOde}. We introduce the following shorthand notation for the flux-derivative coefficients at each $p$-grid point:
\begin{equation}\label{eq:coeffs}
    P_{\alpha,l}^+ = [F'_\alpha(p_l)]_+ \geq 0, \qquad P_{\alpha,l}^- = [F_\alpha'(p_l)]_- \leq 0, \qquad P_{\alpha,l} = P_{\alpha,l}^+ + P_{\alpha,l}^-,
\end{equation}
so that the evolution matrix has the block-diagonal form
\begin{equation}\label{eq:A_structure2}
    \bm{A} = \bigoplus_{l=0}^{N_p-1} \bm{A}^{(l)}, \qquad
    \bm{A}^{(l)} = \sum_{\alpha=1}^{d} \big( P_{\alpha,l}^- D_+^{(\alpha)} + P_{\alpha,l}^+ D_-^{(\alpha)} \big).
\end{equation}
We encode the $p$-index using $n_p = \lceil \log_2 N_p \rceil$ qubits and each spatial dimension using $n_x = \lceil \log_2 N_x \rceil$ qubits, requiring a total of $n_p + d n_x$ qubits. The quantum state encoding the solution is
\begin{equation}\label{eq:quantum_encoding}
    |\bm{\Psi}(t)\rangle = \frac{1}{\|\bm{\Psi}(t)\|} \sum_{l=0}^{N_p-1} \sum_{\bm{j}} \psi_{\bm{j}, l}(t) \, |l\rangle_p \otimes |\bm{j}\rangle_{\bm{x}}.
\end{equation}
By the block-diagonal structure of $\bm{A}$, the matrix acts on this encoded state as $\bm{A}(|l\rangle_p \otimes |\bm{j}\rangle_{\bm{x}}) = |l\rangle_p \otimes \bm{A}^{(l)}|\bm{j}\rangle_{\bm{x}}$. Hence, each $p$-block evolves independently under its corresponding spatial operator.

\subsubsection{Schr\"{o}dingerisation}\label{subsubsec: schroding use}
Since $\bm{A}$ is generally non-Hermitian, we decompose it into its Hermitian and skew-Hermitian components:
\begin{align}\label{eq:A=A1+A2}
    \bm{A} &= \bm{A}_1 + i\bm{A}_2, \nonumber\\
    \bm{A}_{1} &= \frac{\bm{A} + \bm{A}^\dagger}{2} = \frac{1}{2} \bigoplus_{\alpha=1}^d\bigoplus_{l=0}^{N_p-1} (P_{\alpha,l}^- - P_{\alpha,l}^+)(D_+ - D_-),  \\
    \bm{A}_{2} &= \frac{\bm{A} - \bm{A}^\dagger}{2i} = \frac{1}{2i} \bigoplus_{\alpha=1}^d\bigoplus_{l=0}^{N_p-1} (P_{\alpha,l}^- + P_{\alpha,l}^+)(D_+ + D_-), \nonumber
\end{align}
where both $\bm{A}_1$ and $\bm{A}_2$ are Hermitian, so that $i\bm{A}_2$ is skew-Hermitian. Note that $P_{\alpha,l}^- - P_{\alpha,l}^+ = -|P_{\alpha,l}|$.

\par We apply Schr\"{o}dingerisation \cite{SchrodingerJin2022} to embed the non-unitary evolution into a unitary dynamical system. Introducing the warped-phase transformation
\begin{equation}
    \mathbf{v}(t,w) = e^{-w}\bm{\Psi}(t), \qquad w > 0,
\end{equation}
which satisfies $\mathbf{v}(t,0) = \bm{\Psi}(t)$ and $\mathbf{v}(t,w) \to 0$ as $w \to \infty$, a direct calculation gives
\begin{equation*}
    \frac{d\mathbf{v}}{dt} = -\mathbf{A}_1 \partial_w \mathbf{v} + i\mathbf{A}_2 \mathbf{v}.
\end{equation*}
Applying the Fourier transform in $w$ yields the Schr\"{o}dinger-type equation
\begin{equation}
\partial_t \hat{\mathbf{v}}(t,\eta)
= i(\eta \mathbf{A}_1 + \mathbf{A}_2)\hat{\mathbf{v}}(t,\eta)
\triangleq i\mathbf{H}_C \hat{\mathbf{v}}(t,\eta),
\end{equation}
where $\mathbf{H}_C = \eta\mathbf{A}_1 + \mathbf{A}_2$ is a Hermitian operator that serves as the effective Hamiltonian.

\par To discretize the auxiliary variable $w$, we truncate the domain to $w \in [-\pi R, \pi R]$, introduce $N_w$ uniformly spaced grid points, and apply the discrete Fourier transform. The resulting Hamiltonian has the tensor-product structure
\begin{equation}\label{eq:H_tensor}
    \mathbf{H}_C = \mathbf{A}_1 \otimes D_\eta + \mathbf{A}_2 \otimes \mathbf{I}_{N_w},
\end{equation}
where $D_\eta = \mathrm{diag}(\eta_0, \dots, \eta_{N_w-1})$ with $\eta_r = (r-N_w/2)/R$. The system evolves as $d|\hat{\mathbf{v}}(t)\rangle/dt = -i\mathbf{H}_C|\hat{\mathbf{v}}(t)\rangle$, with initial state
\begin{equation}
    \hat{\mathbf{v}}(0) = \mathcal{F}_w[\bm{\Psi}(0)\otimes e^{-|\cdot|}] = \bm{\Psi}(0)\otimes \bm{\eta},
\end{equation}
where $e^{-|\cdot|} = (e^{-|w_0|},\dots,e^{-|w_{N_w-1}|})^\top$ and $\bm{\eta} = \bigl(\frac{2}{\eta_0^2+1},\dots,\frac{2}{\eta_{N_w-1}^2+1}\bigr)^\top$.

\par Substituting the block structure of $\bm{A}_1$ and $\bm{A}_2$, the Hamiltonian can be written explicitly as
\begin{align}\label{HHJ}
\mathbf{H}_C
&= \sum_{l,r} (r-\tfrac{N_w}{2}) |l\rangle\langle l| \otimes \mathbf{H}^{(l)}_1 \otimes |r\rangle\langle r|
+ \sum_l |l\rangle\langle l| \otimes \mathbf{H}^{(l)}_2 \otimes I_{N_w},
\end{align}
where $\mathbf{H}^{(l)}_1 = -\frac{1}{2}\sum_{\alpha=1}^d |P_{\alpha,l}|(\mathbf{H}_1)_\alpha$, $\mathbf{H}^{(l)}_2 = \frac{1}{2i}\sum_{\alpha=1}^d P_{\alpha,l}(\mathbf{H}_2)_\alpha$, $\mathbf{H}_1$ and $\mathbf{H}_2$ are the discrete antisymmetric and symmetric difference operators respectively, and $(H)_\alpha = I^{\otimes(d-\alpha)n_x}\otimes H\otimes I^{\otimes(\alpha-1)n_x}$ denotes the embedding of $H$ into the $\alpha$-th spatial dimension.

\par This tensor-product Hamiltonian structure enables efficient quantum circuit construction via controlled operations, as detailed below.

\subsubsection{Quantum circuits}\label{section quantum circuits}

\par We now construct explicit quantum circuits for the time-evolution operator $U_C(\tau) := \exp(i\mathbf{H}_C\tau)$. Applying the first-order Lie--Trotter--Suzuki decomposition to \eqref{HHJ} gives
\begin{equation}\label{V*}
    \begin{aligned}
        U_C(\tau)
        &\approx \exp\!\Bigl(i\tau \sum_{l}|l\rangle\langle l|\otimes\mathbf{H}^{(l)}_2\otimes I^{\otimes n_w}\Bigr)\exp\!\Bigl(i\tau \sum_{l,r}(r-\tfrac{N_w}{2})|l\rangle\langle l|\otimes \mathbf{H}^{(l)}_1\otimes |r\rangle\langle r|\Bigr)\\
        &=\sum_{l}|l\rangle\langle l|\otimes \mathbf{U}^{(l)}_2(\tau)\otimes I^{\otimes n_w} \cdot \sum_{l,r}|l\rangle\langle l|\otimes \bigl(\mathbf{U}^{(l)}_1(-\tau)\bigr)^{r-N_w/2}\otimes |r\rangle\langle r|
        \;\triangleq\; V_C(\tau),
    \end{aligned}
\end{equation}
where $\mathbf{U}^{(l)}_1(-\tau) = \exp(-i\mathbf{H}^{(l)}_1\tau) = \prod_{\alpha=1}^d \bigl(\exp(-i|P_{\alpha,l}|\mathbf{H}_1\tau)\bigr)_\alpha$ and $\mathbf{U}^{(l)}_2(\tau) = \exp(i\mathbf{H}^{(l)}_2\tau) = \prod_{\alpha=1}^d \bigl(\exp(iP_{\alpha,l}\mathbf{H}_2\tau)\bigr)_\alpha$.

\par We construct each factor explicitly for \textbf{periodic boundary conditions} so that $D^+_P = S^- - I^{\otimes n_x} + \sigma_{10}^{\otimes n_x}$. Other boundary conditions can be incorporated by replacing the corresponding finite difference operators $D^\pm$. The key identity from \cite{Sato2024HamiltonianSF} is
\[
e^{i\lambda}s_j^{-}+e^{-i\lambda}s_j^{+} = I^{\otimes (n_x-j)}\otimes U_j(-\lambda)\bigl(Z\otimes |1\rangle\langle 1|^{\otimes(j-1)}\bigr)U_j(-\lambda)^\dagger, \quad U_j(\lambda)=\Bigl(\prod_{m=1}^{j-1}\mathrm{CNOT}_m^j\Bigr)P_j(\lambda)H_j.
\]

\paragraph{Circuit for $\mathbf{U}_1(-\tau)$.}
With $\gamma_1 = 1/(2hR)$, the Trotter approximation gives
\begin{equation*}
    \exp(-i\mathbf{H}_1\tau) \approx \exp(2i\gamma_1\tau I^{\otimes n_x})\cdot \exp\!\bigl(-i\gamma_1\tau(\sigma_{01}^{\otimes n_x}+\sigma_{10}^{\otimes n_x})\bigr)\cdot \prod_{j=1}^{n_x} \exp\!\bigl(-i\gamma_1\tau(s_j^{-}+s_j^{+})\bigr)
    = Ph(2\gamma_1\tau)\,U_{11}(\tau)\,U_{12}(\tau),
\end{equation*}
where $Ph(\theta) = \exp(i\theta)I^{\otimes n_x}$ denotes a global phase. Using the identity
\[
\sigma_{01}^{\otimes n_x}+\sigma_{10}^{\otimes n_x} = X_{n_x}\bigl(|0\rangle|1\rangle^{\otimes(n_x-1)}\langle 1|\langle 0|^{\otimes(n_x-1)}+|1\rangle|0\rangle^{\otimes(n_x-1)}\langle 0|\langle 1|^{\otimes(n_x-1)}\bigr)X_{n_x},
\]
the corresponding subcircuits are
\begin{equation*}
    \begin{aligned}
        U_{11}(\tau)&=U_{n_x}(0)\,X_{n_x}\,CRZ_{n_x}^{1,\dots,n_x-1}(2\gamma_1\tau)\,X_{n_x}\,U_{n_x}(0)^\dagger,\\
        U_{12}(\tau)&=\prod_{j=1}^{n_x} I^{\otimes (n_x-j)}\otimes W_j(\gamma_1,\tau),
    \end{aligned}
\end{equation*}
where $W_j(\gamma_1,\tau)=U_j(0)\,CRZ_j^{1,\dots,j-1}(-2\gamma_1\tau)\,U_j(0)^\dagger$. This yields
\begin{equation}\label{V1}
    \mathbf{U}_1(-\tau)\approx Ph(2\gamma_1\tau)\,U_{11}(\tau)\,U_{12}(\tau) \triangleq V_1(-\tau).
\end{equation}
The explicit gate-level circuit for $V_1(-\tau)$ is provided in Appendix~\ref{appV1circuits}.

\paragraph{Circuit for $\mathbf{U}_2(\tau)$.}
With $\gamma_2 = 1/(2h)$, the Trotter approximation gives
\begin{equation*}
    \exp(i\mathbf{H}_2\tau)
    \approx \underbrace{\exp\!\bigl(\gamma_2\tau(\sigma_{01}^{\otimes n_x}-\sigma_{10}^{\otimes n_x})\bigr)}_{U_{21}(\tau)}
    \cdot \underbrace{\prod_{j=1}^{n_x}\exp\!\bigl(\gamma_2\tau(s_j^{-}-s_j^{+})\bigr)}_{U_{22}(\tau)}.
\end{equation*}
Using the identities $s_j^{-}-s_j^{+} = -i\bigl(I^{\otimes(n_x-j)}\otimes U_j(-\tfrac{\pi}{2})(Z\otimes|1\rangle\langle1|^{\otimes(j-1)})U_j(-\tfrac{\pi}{2})^\dagger\bigr)$ and
\begin{equation*}
    \sigma_{01}^{\otimes n_x}-\sigma_{10}^{\otimes n_x}
    =-i\, U_{n_x}(-\tfrac{\pi}{2})\,(I\otimes X^{\otimes(n_x-1)})\,(Z\otimes|1\rangle\langle1|^{\otimes(n_x-1)})\,(I\otimes X^{\otimes(n_x-1)})\,U_{n_x}(-\tfrac{\pi}{2})^\dagger,
\end{equation*}
we obtain
\begin{equation*}
    \begin{aligned}
        U_{21}(\tau) &= U_{n_x}(-\tfrac{\pi}{2})\,(I\otimes X^{\otimes(n_x-1)})\,CRZ_{n_x}^{1,\dots,n_x-1}(2\gamma_2\tau)\,(I\otimes X^{\otimes(n_x-1)})\,U_{n_x}(-\tfrac{\pi}{2})^\dagger \;\triangleq\; W'_{n_x}(\gamma_2,\tau),\\
        U_{22}(\tau) &= \prod_{j=1}^{n_x} U_{j}(-\tfrac{\pi}{2})\,(I\otimes X^{\otimes(j-1)})\,CRZ_{j}^{1,\dots,j-1}(2\gamma_2\tau)\,(I\otimes X^{\otimes(j-1)})\,U_{j}(-\tfrac{\pi}{2})^\dagger \;\triangleq\;\prod_{j=1}^{n_x} W'_{j}(\gamma_2,\tau),
    \end{aligned}
\end{equation*}
which gives
\begin{equation}\label{V2}
    \mathbf{U}_2(\tau)\approx U_{21}(\tau)\,U_{22}(\tau) \triangleq V_2(\tau).
\end{equation}
The explicit circuit for $V_2(\tau)$ and its components is provided in Appendix~\ref{appendix2}.

\paragraph{Full circuit for \(V_C(\tau)\).}

Since different spatial directions act on different tensor factors, the factors with different \(\alpha\) commute. Hence, as in \eqref{V*}, for each fixed \(l\),
\[
V_1^{(l)}(-\tau)
:=
\exp(-iH_1^{(l)}\tau)
=
\prod_{\alpha=1}^d
(V_1(|P_{\alpha,l}|\tau))_\alpha,
\]
and
\[
V_2^{(l)}(\tau)
:=
\exp(iH_2^{(l)}\tau)
=
\prod_{\alpha=1}^d
(V_2(P_{\alpha,l}\tau))_\alpha.
\]
Replacing
\(\mathbf U_1^{(l)}\) and \(\mathbf U_2^{(l)}\) by \(V_1^{(l)}\) and
\(V_2^{(l)}\) in \eqref{V*}, we obtain
\begin{equation}\label{VC}
\begin{aligned}
U_C(\tau)
&\approx
\left(
\sum_l
|l\rangle\langle l|
\otimes
V_2^{(l)}(\tau)
\otimes
I^{\otimes n_w}
\right)
\left(
\sum_{l,r}
|l\rangle\langle l|
\otimes
\left(V_1^{(l)}(-\tau)\right)^{r-N_w/2}
\otimes
|r\rangle\langle r|
\right)
\\
&\triangleq
V_C(\tau).
\end{aligned}
\end{equation}
Writing
\[
r=\sum_{m=0}^{n_w-1}r_m2^m,
\qquad
N_w=2^{n_w},
\]
for each fixed \(l\), define
\[
\mathcal V_1^{(l)}[s]
:=
\left(V_1^{(l)}(-\tau)\right)^s
=
\prod_{\alpha=1}^d
(V_1(s|P_{\alpha,l}|\tau))_\alpha,
\qquad
a\in\mathbb Z.
\]
Then the \(w\)-controlled part factorizes as
\begin{equation}\label{Vw_l_definition}
\begin{aligned}
V_w^{(l)}(-\tau)
&:=
\sum_r
\left(V_1^{(l)}(-\tau)\right)^{r-N_w/2}
\otimes
|r\rangle\langle r|
\\
&=
\mathcal V_1^{(l)}[-2^{n_w-1}]
\prod_{m=0}^{n_w-1}
\left[
I^{\otimes dn_x}\otimes |0\rangle\langle0|_{r_m}
+
\mathcal V_1^{(l)}[2^m]
\otimes |1\rangle\langle1|_{r_m}
\right].
\end{aligned}
\end{equation}
Therefore,
\[
V_C(\tau)
=
\left(
\sum_l
|l\rangle\langle l|
\otimes
V_2^{(l)}(\tau)
\otimes
I^{\otimes n_w}
\right)
\left(
\sum_l
|l\rangle\langle l|
\otimes
V_w^{(l)}(-\tau)
\right).
\]
The circuit for the fixed-\(l\) block \(V_w^{(l)}(-\tau)\) is
\[
\Qcircuit @C=1em @R=1.3em {
\lstick{q}
& \qw {/}^{dn_x}
& \gate{\mathcal V_1^{(l)}[2^0]}
& \gate{\mathcal V_1^{(l)}[2^1]}
& \qw
& \cdots
& 
& \gate{\mathcal V_1^{(l)}[2^{n_w-1}]}
& \gate{\mathcal V_1^{(l)}[-2^{n_w-1}]}
& \qw
\\
\lstick{r_0}
& \qw
& \ctrl{-1}
& \qw
& \qw
& \cdots
& 
& \qw
& \qw
& \qw
\\
\lstick{r_1}
& \qw
& \qw
& \ctrl{-2}
& \qw
& \cdots
& 
& \qw
& \qw
& \qw
\\
\lstick{\vdots}
& \qw
& \qw
& \qw
& \qw
& \ddots
& 
& \qw
& \qw
& \qw
\\
\lstick{r_{n_w-1}}
& \qw
& \qw
& \qw
& \qw
& \cdots
& 
& \ctrl{-4}
& \qw
& \qw
}
\]

\begin{remark}[Binary-controlled implementation for the linear case]\label{rem: binary}
When $P_{\alpha,l}$ is affine in the $p$-grid index, the multiplexed operator $\ket{l}\bra{l}\otimes V^{(l)}$ can be implemented using binary-controlled powers, in the same manner as the $w$-controlled construction. Further details are provided in Appendix~\ref{appendix:remark}.
\end{remark}

\begin{remark}[General nonlinear flux derivatives]
\label{rem:nonlinear-flux}
For a general nonlinear flux derivative, the corresponding
block-diagonal evolution can be implemented directly as the
multiplexed unitary $V=\sum_{l=0}^{N_p-1}|l\rangle\langle l|_p\otimes V^{(l)},$ where the value \(P_{\alpha,l}\) is encoded in the evolution parameters of \(V^{(l)}\). The \(p\)-register itself therefore selects the appropriate unitary. A direct implementation applies a multi-controlled version of
\(V^{(l)}\) for each \(l\). In the absence of additional structure in \(F'_\alpha\), this generally introduces a cost linear in \(N_p\), in contrast to the logarithmic-depth binary-power construction available for affine flux derivatives. More structured nonlinear functions may admit more efficient arithmetic-based or polynomial-approximation implementations.
\end{remark}

\par To evolve the system to time $T$, we apply $V_C(\tau)$ a total of $T/\tau$ times. The complete circuit is given by
\begin{equation}\label{endcircuit}
    \Qcircuit @C=1.5em @R=1.5em {
   & \lstick{\ket{\bm{\Psi}(0)}} &\qw {/}^{dn_x+n_p}\qw&\qw&\multigate{1}{V_C^{T/\tau}(\tau)}&\qw&\qw&\\
   & \lstick{\ket{w}}&\qw {/}^{n_w}\qw &\gate{\mathcal{F}} &\ghost{V_C^{T/\tau}(\tau)}&\gate{\mathcal{F}^\dagger}&\qw&
}
\end{equation}
where $|\bm{\Psi}(0)\rangle = \sum_{l,\bm{j}} \psi_{\bm{j},l}(0)|l\rangle|\bm{j}\rangle$ is the encoded initial state \eqref{eq:quantum_encoding}, and $\mathcal{F}$ and $\mathcal{F}^\dagger$ denote the QFT and inverse QFT on the $w$-register, respectively.

\subsection{Solution recovery and computing physical observables}
\label{subsec: solution recovery}
\label{subsec:observables}

\par After the quantum circuit \eqref{endcircuit} prepares the final state
$|\bm{\Psi}(T)\rangle$, the quantum state encoding the evolved
distribution function is
\begin{equation}\label{eq:final_state}
 |\bm{\Psi}(T)\rangle
 =
 \frac{1}{\mathcal{N}_T}
 \sum_{l=0}^{N_p-1}
 \sum_{\bm{j}}
 \psi_{\bm{j},l}(T)
 |l\rangle_p\otimes|\bm{j}\rangle_{\bm{x}},
\end{equation}
where $\mathcal{N}_T^2 =\sum_{\bm{j},l} |\psi_{\bm{j},l}(T)|^2.$ Physical quantities are extracted by computing moments of $\psi$ with
respect to the auxiliary variable $p$.

\par Given a test function
$G:\mathbb{R}\rightarrow\mathbb{R}$, the corresponding physical
observable is defined as
\begin{equation}\label{eq:obs_def}
 \langle G\rangle(t,\bm{x})
 :=
 \int_{\mathbb{R}}
 G(p)\psi(t,\bm{x},p)\,dp,
\end{equation}
which, after discretization and mollification of the delta function,
becomes
\begin{equation}\label{eq:obs_disc}
 \langle G\rangle_{\bm{j}}(T)
 :=
 \frac{1}{N_p}
 \sum_{l=0}^{N_p-1}
 G(p_l)\psi_{\bm{j},l}(T).
\end{equation}
For single-valued solutions,
$\langle G\rangle_{\bm{j}}\approx
G(u(T,\bm{x}_{\bm{j}}))$.
For multi-valued solutions, this observable represents the Jacobian-weighted ensemble
average over all characteristic branches. We next describe four
important classes of observables.
\subsubsection*{Zeroth moment: $G(p)=1$}

The zeroth moment
\begin{equation}
    \langle 1\rangle(t,\bm{x})
 =
 \int_{\mathbb{R}}\psi(t,\bm{x},p)\,dp
\end{equation}
represents the total density in the $p$ variable. Applying $H^{\otimes n_p}$ to the
$p$-register gives
\begin{equation}\label{eq:hadamard_p}
 H^{\otimes n_p}|l\rangle
 =
 \frac{1}{\sqrt{N_p}}
 \sum_{k=0}^{N_p-1}
 (-1)^{l\cdot k}|k\rangle.
\end{equation}
Therefore, the component associated with
$|k\rangle=|0\rangle^{\otimes n_p}$ is
\begin{equation}
 \frac{1}{\mathcal{N}_T\sqrt{N_p}}
 \sum_{\bm{j}}
 \left(
   \sum_{l=0}^{N_p-1}
   \psi_{\bm{j},l}(T)
 \right)
 |\bm{j}\rangle_{\bm{x}}
 |0\rangle_p^{\otimes n_p}.
\end{equation}
Post-selecting the all-zero state of the $p$-register therefore yields
\begin{equation}\label{eq:zeroth-recovered-state}
 |\psi_{\mathrm{rec}}^{(0)}\rangle
 \propto
 \sum_{\bm{j}}
 \left(
   \sum_{l=0}^{N_p-1}
   \psi_{\bm{j},l}(T)
 \right)
 |\bm{j}\rangle_{\bm{x}},
\end{equation}
which encodes the discrete zeroth moment up to known normalization and
quadrature factors. The corresponding circuit is
\begin{equation}\label{eq:zeroth-moment-circuit}
\Qcircuit @C=1.0em @R=0.5em {
\lstick{l_0}
 & \multigate{3}{H^{\otimes n_p}}
 & \meter
 & \rstick{|0\rangle}
\\
\lstick{l_1}
 & \ghost{H^{\otimes n_p}}
 & \meter
 & \rstick{|0\rangle}
\\
\lstick{\vdots}
 & \ghost{H^{\otimes n_p}}
 & \vdots
 &
\\
\lstick{l_{n_p-1}}
 & \ghost{H^{\otimes n_p}}
 & \meter
 & \rstick{|0\rangle}
\\
\lstick{q}
 & \qw
 & \meter
 &
}
\end{equation}

\subsubsection*{Affine observables and solution recovery:
$G(p)=ap+b$.}

The discrete affine observable is
\begin{equation}\label{eq:affine-observable}
 \langle ap+b\rangle_{\bm{j}}
 =
 \frac{1}{N_p}
 \sum_{l=0}^{N_p-1}
 (ap_l+b)\psi_{\bm{j},l}(T).
\end{equation}
Let
\begin{equation}\label{eq:affine-binary}
    ap_l+b =
 (ap_{\min}+b)
 +
 ah_p\sum_{m=0}^{n_p-1}l_m2^m,\quad\quad
 p_l=p_{\min}+lh_p,
 \quad
 l=\sum_{m=0}^{n_p-1}l_m2^m,
\end{equation}
Thus, an affine observable can be encoded directly from the binary
digits of the $p$-register without introducing an additional value
oracle.

Let
\[
 M_{a,b}
 :=
 \max_{0\leq l<N_p}|ap_l+b|
\]
and choose a sufficiently small scaling parameter $\theta_0$. Define
the offset and bit-dependent rotation angles by
\begin{equation}\label{eq:affine-angles}
 \theta_{\mathrm{off}}
 =
 \frac{\theta_0(ap_{\min}+b)}{M_{a,b}},
 \qquad
 \theta_m^{(a)}
 =
 \frac{\theta_0a2^mh_p}{M_{a,b}},
 \quad
 m=0,\ldots,n_p-1.
\end{equation}

\par\paragraph{Step 1: Affine encoding.}
Starting from an ancilla qubit $|0\rangle_a$, we apply the unconditional
rotation $R_y(2\theta_{\mathrm{off}})$, followed by
$R_y(2\theta_m^{(a)})$ controlled by the $m$-th qubit of the
$p$-register. For a basis state $|l\rangle_p$, the accumulated rotation angle is
\begin{equation}\label{eq:affine-total-angle}
 \Theta_l^{(a,b)}
 =
 \theta_{\mathrm{off}}
 +
 \sum_{m=0}^{n_p-1}
 l_m\theta_m^{(a)}
 =
 \frac{\theta_0(ap_l+b)}{M_{a,b}}.
\end{equation}
The resulting transformation is
\begin{equation}\label{eq:affine-ancilla}
 |l\rangle_p|0\rangle_a
 \longmapsto
 |l\rangle_p
 \left[
   \cos\bigl(\Theta_l^{(a,b)}\bigr)|0\rangle_a
   +
   \sin\bigl(\Theta_l^{(a,b)}\bigr)|1\rangle_a
 \right].
\end{equation}
In the small-angle regime,
\begin{equation}
 \sin\bigl(\Theta_l^{(a,b)}\bigr)
 =
 \frac{\theta_0(ap_l+b)}{M_{a,b}}
 +
 \mathcal{O}(\theta_0^3).
\end{equation}

\par\paragraph{Step 2: Hadamard summation.}
We apply $H^{\otimes n_p}$ to the $p$-register. The
$|0\rangle^{\otimes n_p}$ component then contains
\[
 \frac{1}{\sqrt{N_p}}
 \sum_{l=0}^{N_p-1}
 \sin\bigl(\Theta_l^{(a,b)}\bigr)
 \psi_{\bm{j},l}(T)
\]
in the $\bm{x}$-register.

\par\paragraph{Step 3: Post-selection.}
Post-selecting $|0\rangle^{\otimes n_p}$ on the $p$-register and
$|1\rangle_a$ on the ancilla yields
\begin{equation}\label{eq:affine-recovery}
 |\psi_{\mathrm{rec}}^{(a,b)}\rangle
 \propto{}
 \frac{\theta_0\sqrt{N_p}}{M_{a,b}}
 \sum_{\bm{j}}
 \langle ap+b\rangle_{\bm{j}}
 |\bm{j}\rangle_{\bm{x}}+
 \mathcal{O}(\theta_0^3).
\end{equation}

Solution recovery corresponds to the particular affine observable obtained by
setting
\[
 a=1,
 \qquad
 b=0,
 \qquad
 G(p)=p.
\]

The affine circuit, which also gives the solution-recovery circuit when
$(a,b)=(1,0)$, is
\begin{equation}\label{eq:recovery_circuit}
\Qcircuit @C=0.8em @R=0.4em {
\lstick{l_0}
 & \qw
 & \ctrl{4}
 & \qw
 & \qw
 & \qw
 & \multigate{3}{H^{\otimes n_p}}
 & \meter
 & \rstick{|0\rangle}
\\
\lstick{l_1}
 & \qw
 & \qw
 & \ctrl{3}
 & \qw
 & \qw
 & \ghost{H^{\otimes n_p}}
 & \meter
 & \rstick{|0\rangle}
\\
\lstick{\vdots}
 & \qw
 & \qw
 & \qw
 & \ddots
 & \qw
 & \ghost{H^{\otimes n_p}}
 & \vdots
 &
\\
\lstick{l_{n_p-1}}
 & \qw
 & \qw
 & \qw
 & \qw
 & \ctrl{1}
 & \ghost{H^{\otimes n_p}}
 & \meter
 & \rstick{|0\rangle}
\\
\lstick{|0\rangle_a}
 & \gate{R_y(2\theta_{\mathrm{off}})}
 & \gate{R_y(2\theta_0^{(a)})}
 & \gate{R_y(2\theta_1^{(a)})}
 & \cdots
 & \gate{R_y(2\theta_{n_p-1}^{(a)})}
 & \qw
 & \meter
 & \rstick{|1\rangle}
\\
\lstick{q}
 & \qw
 & \qw
 & \qw
 & \qw
 & \qw
 & \qw
 & \meter
 &
}
\end{equation}

\subsubsection*{General nonlinear observables: $G(p)=\eta(p)$}

For a general nonlinear function $\eta$, such as a convex entropy
function, the values $\eta(p_l)$ cannot, in general, be represented as a
linear combination of the binary digits of $l$. We therefore introduce
a coherent value oracle. Let
\begin{equation}
    \eta_l := \eta(p_l),
    \qquad
    M_\eta \geq \max_{0\leq l<N_p}|\eta_l|,
    \qquad
    \bar{\eta}_l := \frac{\eta_l}{M_\eta}\in[-1,1].
    \label{eq:eta-normalization}
\end{equation}
Using an $n_\eta$-qubit value register, the oracle acts as
\begin{equation}
    O_\eta:
    |l\rangle_p|0\rangle_\eta^{\otimes n_\eta}
    \longmapsto
    |l\rangle_p|\widetilde{\eta}_l\rangle_\eta,
    \label{eq:eta-oracle}
\end{equation}
where $\widetilde{\eta}_l$ is a fixed-point approximation of
$\bar{\eta}_l$. The oracle can be implemented either by a coherent lookup of
classically precomputed values $\eta_l$ or by reversible arithmetic
for evaluating $\eta(p_l)$.

Writing the oracle output in signed binary form,
\begin{equation}
    \widetilde{\eta}_l
    =
    (-1)^{s_l}
    \sum_{r=1}^{n_\eta-1}
    \eta_{l,r}2^{-r},
    \qquad
    s_l,\eta_{l,r}\in\{0,1\},
    \label{eq:eta-binary}
\end{equation}
the bits $\eta_{l,r}$ control rotations with angles
$\theta_0 2^{-r}$ on the ancillary qubit, while the sign bit $s_l$
determines the rotation direction. The accumulated angle is therefore
\begin{equation}
    \Theta_l^{(\eta)}
    =
    \theta_0\widetilde{\eta}_l.
    \label{eq:eta-total-angle}
\end{equation}
Consequently,
\begin{equation}
|l\rangle_p|\widetilde{\eta}_l\rangle_\eta|0\rangle_a
\rightarrow
|l\rangle_p|\widetilde{\eta}_l\rangle_\eta\otimes
\left[
\cos\bigl(\Theta_l^{(\eta)}\bigr)|0\rangle_a
+
\sin\bigl(\Theta_l^{(\eta)}\bigr)|1\rangle_a
\right].
\label{eq:eta-controlled-rotation}
\end{equation}

The oracle is then uncomputed by applying $O_\eta^\dagger$. The
remaining Hadamard summation and post-selection steps are identical to those
used for the affine observable. 

\subsubsection*{Ensemble averages over multiple initial data.}

For $M$ initial conditions
$\{u_0^{[k]}\}_{k=1}^{M}$, define the ensemble-averaged initial
distribution by
\begin{equation}
    \psi_{0,\bm{j},l}^{\mathrm{ens}}
    :=
    \frac{1}{M}
    \sum_{k=1}^{M}
    \delta_\omega
    \bigl(
        p_l-u_0^{[k]}(\bm{x}_{\bm{j}})
    \bigr).
    \label{eq:ensemble-initial-distribution}
\end{equation}
The corresponding initial quantum state is
\begin{equation}
\begin{aligned}
|\bm{\Psi}_0^{\mathrm{ens}}\rangle
=
\frac{1}{\mathcal{N}_{\mathrm{ens}}}
\sum_{l,\bm{j}}
\psi_{0,\bm{j},l}^{\mathrm{ens}}
|l\rangle_p|\bm{j}\rangle_{\bm{x}},
\end{aligned}
\label{eq:ensemble-initial-state}
\end{equation}

Since the evolution equation for $\psi$ is linear, evolving
$\psi_0^{\mathrm{ens}}$ yields
\begin{equation}
    \psi^{\mathrm{ens}}(T,\bm{x},p)
    =
    \frac{1}{M}
    \sum_{k=1}^{M}
    \psi^{[k]}(T,\bm{x},p),
    \label{eq:ensemble-evolved-distribution}
\end{equation}
where $\psi^{[k]}$ denotes the distribution evolved from the initial
condition $u_0^{[k]}$. Applying any of the preceding observable
protocols therefore gives
\begin{equation}
\begin{aligned}
\langle G\rangle_{\mathrm{ens}}(T,\bm{x})
&=
\int_{\mathbb{R}}
G(p)\psi^{\mathrm{ens}}(T,\bm{x},p)\,dp
\\
&=
\frac{1}{M}
\sum_{k=1}^{M}
\langle G\rangle^{[k]}(T,\bm{x}).
\end{aligned}
\label{eq:ensemble-average}
\end{equation}
A label register is required only if observables associated with individual initial conditions need to be retained or accessed separately. Once the ensemble-averaged initial state
\eqref{eq:ensemble-initial-state} has been prepared, the subsequent
Hamiltonian-simulation and observable-estimation circuits are
independent of $M$.

\begin{lem}[Observable discretization error]
\label{lemma_epsilon_G}
Let \(G\in W^{2,\infty}(\mathbb R)\), and let
\(\langle G\rangle_{\bm j,h}(T)\) be the discrete approximation of
\(\langle G\rangle(T,\bm x_{\bm j})\), obtained using a symmetric
mollifier of width \(\omega\), a \(p\)-grid spacing \(h_p\), and a
first-order spatial discretization with spacing \(h\). Then
\begin{equation}
\label{eq:epsilon_omega}
\left|
\langle G\rangle(T,\bm x_{\bm j})
-
\langle G\rangle_{\bm j,h}(T)
\right|
\leq
C\left(
\omega^2+
\frac{h_p^2}{\omega^2}
+
\frac{dhT}{\omega^2}
\right).
\end{equation}
\end{lem}

\begin{proof}
See Appendix~\ref{appendix3}.
\end{proof}

For \(h_p\sim h\), balancing the dominant terms in
\eqref{eq:epsilon_omega} gives
\[
\omega=O(\varepsilon^{1/2}),
\qquad
h=O\left(\frac{\varepsilon^2}{dT}\right),
\qquad
N_x=O\left(dT\varepsilon^{-2}\right).
\]

To improve this dependence, we cancel the leading mollification error
by Richardson extrapolation.

\begin{cor}[Richardson-corrected kernel]
\label{cor: richardson_delta}
Assume additionally that \(G\in W^{4,\infty}(\mathbb R)\) and that the
solution is sufficiently smooth. For
\[
\delta_\omega^{R}
=
\frac{4\delta_{\omega/2}-\delta_\omega}{3},
\]
the corrected observable satisfies
\begin{equation}
\label{eq:richardson-observable-error}
\left|
\langle G\rangle(T,\bm x_{\bm j})
-
\langle G\rangle_{\bm j,h}^{R}(T)
\right|
\leq
C\left(
\omega^4+
\frac{h_p^2}{\omega^2}
+
\frac{dhT}{\omega^2}
\right).
\end{equation}
\end{cor}

For \(h_p\sim h\), taking
\(\omega=O(\varepsilon^{1/4})\) and
\(h=O(\varepsilon^{3/2}/(dT))\) gives
\(N_x=O(dT\varepsilon^{-3/2})\).

\section{Error and Complexity Analysis}
\label{sec:error}

\par This section establishes rigorous error bounds and gate complexity estimates for the quantum algorithm. The dominant source of error is the Lie--Trotter--Suzuki decomposition used in the Hamiltonian simulation, with additional contributions from the finite difference discretization (controlled by $h$ and $n_x$) and the Schr\"{o}dingerisation of the $w$-variable (controlled by $n_w$).

\subsection{Trotter error}
\par Our operators $U_1(-\tau)$, $V_1(-\tau)$, $U_2(\tau)$, $V_2(\tau)$ for a single spatial dimension $\alpha$ are structurally analogous to those studied in \cite{Hu2024QuantumCF}, Lemmas~4 and~7. Define the scaled flux coefficients
\begin{equation}\label{eq:beta_def}
    \beta_{1,\alpha,l} := \frac{P_{\alpha,l}^- - P_{\alpha,l}^+}{2h},\quad \beta_{2,\alpha,l} :=\frac{|P_{\alpha,l}^- + P_{\alpha,l}^+|}{2h}\quad |\beta_{1,\alpha,l}|=|\beta_{2,\alpha,l}| = \frac{|F'_\alpha(p_l)|}{2h},
\end{equation}
and let $\beta_{\max} := \max_{\alpha,l}(\beta_{2,\alpha,l})$. For a fixed $l$ the Hamiltonian decomposes as $H_C^{(l)} = \sum_{\alpha=1}^{d}(H_{1,\alpha}^{(l)} + H_{2,\alpha}^{(l)})$ where
\begin{align}
    H_{1,\alpha}^{(l)} &=  \underbrace{\beta_{1,\alpha,l} (S_-^{(\alpha)} - S_+^{(\alpha)})}_{A_1^{(\alpha,l)}} \otimes \; D_w, \\
    H_{2,\alpha}^{(l)} &= -i\underbrace{\beta_{2,\alpha,l} (S_-^{(\alpha)} + S_+^{(\alpha)} - 2I)}_{A_2^{(\alpha,l)}} \otimes \; I_{N_w}\\
    U_\alpha^{(l)} &:= e^{i(H^{(l)}_{1,\alpha}+H^{(l)}_{2,\alpha})\tau} ,\quad \tilde{U}_\alpha^{(l)} := e^{i(H^{(l)}_{2,\alpha})\tau}\cdot e^{i(H^{(l)}_{1,\alpha})\tau}.
\end{align}\par
The Trotter errors for $U_{1,\alpha}^{(l)}$ and $U_{2,\alpha}^{(l)}$ satisfy \cite{Hu2024QuantumCF}:
\begin{align}
    \|U_{1,\alpha}^{(l)}(-\tau) - V_{1,\alpha}^{(l)}(-\tau)\| &\leq \beta_{1,\alpha,l}^2 \tau^2 n_x\cdot \frac{1}{R^2},\label{eq:hu_U1V1}  \\
    \|U_{2,\alpha}^{(l)}(\tau) - V_{2,\alpha}^{(l)}(\tau)\| &\leq \beta_{2,\alpha,l}^2 \tau^2 n_x.\label{eq:hu_U2V2}
\end{align}

\par 
Since operators for different spatial dimensions act on disjoint qubit registers, $[H_{*,\alpha}^{(l)}, H_{*,\beta}^{(l)}] = 0$ for $\alpha \neq \beta$. The only nonzero commutators are within the same dimension $\alpha$. By the first-order Trotter formula \cite{trotter1}:
\begin{equation}
     \|\tilde{U}_\alpha^{(l)}(\tau)- U_\alpha^{(l)}(\tau)\| \leq \frac{\tau^2}{2} \|[H_{1,\alpha}^{(l)}, H_{2,\alpha}^{(l)}]\|.
\end{equation}
Since $[S_-^{(\alpha)}, S_+^{(\alpha)}] = \sum_{j=1}^{n_x} [s_{j,\alpha}^-, s_{j,\alpha}^+]$, $\|[s_{j}^-, s_{j}^+]\| = 1$, and $\|D_w\| = 1/R$, the commutator norm satisfies
\begin{equation}
    \|[H_{1,\alpha}^{(l)}, H_{2,\alpha}^{(l)}]\| \leq 2\beta_{1,\alpha,l}\,\beta_{2,\alpha,l}\, n_x \|D_w\| \leq 2\beta_{\max}^2 n_x R.
\end{equation}
Summing over dimensions and incorporating the per-operator bounds \eqref{eq:hu_U1V1}--\eqref{eq:hu_U2V2}:
\begin{align}
    \|e^{i H_C^{(l)} \tau} - V_C^{(l)}(\tau)\|&=d\cdot \|U_{\alpha}^{(l)}(\tau)-V_{\alpha}^{(l)}(\tau)\|\\
    &\leq d\cdot\left(\|U_{\alpha}^{(l)}(\tau)-\tilde{U}_{\alpha}^{(l)}(\tau)\|+\|\tilde{U}_{\alpha}^{(l)}(\tau)-V_{\alpha}^{(l)}(\tau)\|\right)\\
    &\leq d\cdot\left(\frac{\tau^2}{2}\|H_{1,\alpha}^{(l)},H_{2,\alpha}^{(l)}\|+N_w\cdot\|U_{1,\alpha}^{(l)}(-\tau) - V_{1,\alpha}^{(l)}(-\tau)\|+\|U_{2,\alpha}^{(l)}(\tau) - V_{2,\alpha}^{(l)}(\tau)\|\right) \\
    &\leq d\tau^2\beta^2_{\max} n_x\frac{1}{R}+ d\tau^2\beta^2_{\max}n_x\frac{N_w}{R^2}+d\tau^2\beta^2_{\max} n_x\\
    &\leq 3d\tau^2\beta^2_{\max}n_x\frac{N_w}{R^2}
\end{align}
Since $R=O(1)$ and $N_w\gg 1 $, the $N_w/R^2$ dominates. Finally, both the exact evolution and its circuit approximation are block diagonal in the \(p\)-register. Hence,
\[
\left\|
U_C(\tau)-V_C(\tau)
\right\|
=
\max_l
\left\|
e^{iH_C^{(l)}\tau}
-
V_C^{(l)}(\tau)
\right\|,
\]
so no additional factor \(N_p\) arises from combining the \(p\)-blocks.

\begin{thm}\label{thm:trotter_error}
Consider the Schr\"{o}dinger equation $\dot{\mathbf{u}}(t) = i\mathbf{H}_C\mathbf{u}(t)$ with Hamiltonian $\mathbf{H}_C$ given by \eqref{HHJ}. The time evolution operator $U_C(\tau) = \exp(i\mathbf{H}_C\tau)$ is approximated by $V_C(\tau)$ in \eqref{VC} with
\begin{equation*}
    \|U_C(\tau)-V_C(\tau)\| \leq 3d\tau^2\beta^2_{\max}n_x\frac{N_w}{R^2},\qquad \beta_{\max}=\frac{\max_{\alpha,l}|F'_\alpha(p_l)|}{2h}.
\end{equation*}\par 
With $r=T/\tau$ steps and first order Trotter error accumulating linearly:$\varepsilon_{\tau}\leq r\cdot \|U_C(\tau)-V_C(\tau)\|$. Solving for $\tau$ and $r$:
\begin{equation}\label{eq:r_count}
    \tau = \frac{\varepsilon_\tau}{3dn_xT\beta^2_{\max}N_w/R^2},\quad r=\frac{3dn_xT\beta^2_{\max}N_w/R^2}{\varepsilon_\tau}
\end{equation}
\end{thm}

\subsection{Complexity analysis}
\par The gate complexity of the full algorithm is given by the following lemmas:

\begin{lem}[Per-step gate count]\label{lem:single_step_gates}
The operator $V_C(\tau)$ in \eqref{VC} uses $Q_1=\mathcal{O}(d N_p N_w n_x)$ single-qubit gates and at most $Q_2=\mathcal{O}(d N_p N_w n_x^2)$ CNOT gates, for $n_x \geq 3$.
\end{lem}
\begin{proof}
    Each $V_{1,\alpha}$($V_{2,\alpha}$) decomposes via multi-controlled $R_Z$\cite{Hu2024QuantumCF}(16m-40 CNOT per m-control): $\mathcal{O}(n_x)$ single-qubit, $\mathcal{O}(n_x^2)$ CNOT. With $N_p$ blocks, $N_w$ factor form $w$-controlled $V_1$ and $d$ spatial dimensions, there's $\mathcal{O}(d N_p N_w n_x)$ single-qubit gates and $\mathcal{O}(d N_p N_w n_x^2)$ CNOT gates. 
\end{proof}
\begin{lem}[Evolution gate count]\label{lem: gate complexity}
The time evolution operator $U_C(T) = \exp(i\mathbf{H}_C T)$ on an $(n_p + d n_x + n_w)$-qubit system can be implemented within additive error $\varepsilon_\tau$ using
\[
Q_{1,evol}=\mathcal{O}\!\left( \frac{d^2 n_x^2 T^2 N_p N_w^2\beta_{\max}^2}{R^2\varepsilon_\tau} \right)
\]
single-qubit gates and
\[
Q_{2,evol}=
\mathcal{O}\!\left( \frac{d^2 n_x^3 T^2 N_p N_w^2\beta_{\max}^2 }{R^2\varepsilon_\tau} \right)
\]
CNOT gates.
\end{lem}
\begin{proof}
    From ~\eqref{eq:r_count}: $r=3dn_xT^2\beta^2N_w/(R^2\varepsilon_\tau)$. From Lemma \ref{lem:single_step_gates}: 
    \begin{align*}
        Q_{1,evol}=r\cdot Q_1 &= \mathcal{O}\left(dn_xT^2\beta_{\max}^2N_w/(R^2\varepsilon_\tau)\cdot dN_pN_wn_x\right)\\
        &=\mathcal{O}(d^2n_x^2T^2\beta_{\max}^2 N_p N_w^2/(R^2\varepsilon_\tau))\\
        Q_{2,evol}=r\cdot Q_2 &= \mathcal{O}\left(dn_xT^2\beta_{\max}^2N_w/(R^2\varepsilon_\tau)\cdot dN_pN_wn_x^2\right)\\
        &=\mathcal{O}(d^2n_x^3T^2\beta_{\max}^2 N_p N_w^2/(R^2\varepsilon_\tau))
    \end{align*}
\end{proof}

\begin{thm}[Total complexity with Richardson-corrected mollification]\label{thm:totalcomplexity}
The evolution circuit plus solution recovery estimates
\(\langle G\rangle(T,x)\) to precision \(\varepsilon\) using at most
\[
Q_{\rm total}^{(1q)}
=
\widetilde O
\left(
d^2T^2N_pN_w^2h^{-2}\varepsilon^{-1}
\right)
\]
single-qubit gates, and
\[
Q_{\rm total}^{(2q)}
=
\widetilde O
\left(
d^2T^2N_pN_w^2h^{-2}\varepsilon^{-1}
\right)
\]
CNOT gates.

Using the Richardson-corrected mollifier,
\[
h^{-2}
=
O\left(
d^2T^2\varepsilon^{-3}
\right),
\]
Hence
\[
Q_{\rm total}^{(1q)}
=
\widetilde O
\left(
d^4T^4N_pN_w^2\varepsilon^{-4}
\right),
\]
\[
Q_{\rm total}^{(2q)}
=
\widetilde O
\left(
d^4T^4N_pN_w^2\varepsilon^{-4}
\right).
\]
\end{thm}
\begin{proof}
    See Appendix \ref{appendix4}.
\end{proof}

\begin{remark}[Improved \(w\)-discretization]\label{remarkofw}
The \(w\)-variable discretization corresponds to a Fourier approximation of the profile \(e^{-|w|}\) on \([-\pi R,\pi R]\). Since \(e^{-|w|}\) is continuous but not \(C^1\) at \(w=0\), the Fourier discretization error is only first order, namely
\[
\varepsilon_w = O(N_w^{-1}),\quad N_w=O(\varepsilon^{-1}).
\]
This dependence can be improved by replacing \(e^{-|w|}\) near \(w=0\) with a smoother function \(g(w)\) while keeping the same exponential decay away from the origin\cite{jin2024schrodingerisationill-posed}. If \(g\in C^k\), then the Fourier approximation error improves to
\[
\varepsilon_w = O(N_w^{-(k+1)}),\quad
N_w=O(\varepsilon^{-1/(k+1)}).
\]
Substituting this into Theorem \ref{thm:totalcomplexity} gives
\[
Q_{\rm total}
=
\widetilde O\left(
d^4T^4N_p
\varepsilon^{-3-\frac{2}{k+1}}
\right).
\]
For instance, taking $g(w) = \bigl(\tfrac{3}{e}-3\bigr)w^3+\bigl(\tfrac{4}{e}-5\bigr)w^2-w+1 \in C^1(\mathbb{R})$ on $[-\pi R, 0]$ yields $k=1$ and $\mathcal{O}(\Delta w^2)$ convergence and hence
\[
Q_{\rm total}
=
\widetilde O\left(
d^4T^4N_p\varepsilon^{-4}
\right).
\]
\end{remark}
\begin{remark}[Linear case of $P_{\alpha,l}$]
There's advantage of the binary implementation in remark \ref{rem: binary} on $N_p$. In the linear case
\(P_{\alpha,l}=l\), \(l=\sum_k l_k2^k\), the \(p\)-dependence can be implemented by \(l_k\)-controlled powers. The \(V_2\)-part requires \(O(n_p)\) controlled
blocks, while the \(V_w\)-part requires \(O(n_p n_w)\) singly or doubly controlled blocks. Thus the \(N_p\)-dependent multiplexing overhead is reduced to
a polylogarithmic overhead in \(N_p\), up to the cost of implementing the controlled one-dimensional blocks.
\end{remark}

\begin{remark}[Quantum advantages]
\label{rem:classical-comparison}

Following the observable-computation framework of
Jin and Liu~\cite{Jin2022QuantumAF}, consider a first-order classical
finite-difference method on a \(d\)-dimensional spatial grid with
\(N_x\) points in each direction. Each time step requires
\(O(dN_x^d)\) operations, while the CFL condition
\(\Delta t=O(h/d)\) requires \(O(TdN_x)\) time steps. Hence,
\begin{equation}
\label{eq:classical-complexity}
C_{\mathrm{cl}}
=
O\left(Td^2N_x^{d+1}\right).
\end{equation}
For first-order accuracy, taking
\(N_x=O(dT/\varepsilon)\) gives
\begin{equation}
C_{\mathrm{cl}}
=
O\left(
d^{d+3}T^{d+2}\varepsilon^{-(d+1)}
\right).
\end{equation}

By contrast, under the parameter choices of
Theorem~\ref{thm:totalcomplexity} with the Richardson-corrected scheme requires $N_p=\Theta(\varepsilon^{-3/4})$, the quantum gate complexity for
estimating the observable is
\[
Q_{\mathrm{total}}
=
\widetilde O\left(
d^4T^4\varepsilon^{-27/4}
\right), \quad\text{or}\quad \widetilde O\left(
d^4T^4\varepsilon^{-4}
\right)\quad \text{when}\quad P_{\alpha,l}=l
\]
The classical method explicitly evolves \(N_x^d\) spatial values,
whereas the quantum algorithm represents the spatial grid using
\(dn_x\) qubits and extracts the observable coherently. The comparison
therefore indicates a potential advantage in high spatial dimensions.

This advantage applies to observable estimation rather than complete
solution reconstruction. Reading out all \(N_x^d\) solution values
would introduce a correspondingly large measurement cost. Efficient
state preparation and implementation of the required coefficient
oracles are also assumed.
\end{remark}

\section{Numerical Experiments}
\label{sec:numerics}

\hspace{1.5em} We present numerical experiments to validate the proposed Schr\"odingerisation-based quantum algorithm for scalar conservation laws. The experiments serve three purposes. First, we test the one-dimensional nonlinear transport dynamics using a traffic-flow flux. Second, we demonstrate that the formulation extends naturally to two spatial dimensions by considering a two-dimensional Burgers equation. Third, we examine the gate-count scaling and the convergence with respect to the auxiliary Fourier resolution \(N_w=2^{n_w}\). The quantum circuits and gate-count evaluations were implemented using the open-source UnitaryLab library
\cite{UnitaryLab2026}.

Throughout this section, we use periodic boundary conditions. In one spatial dimension, we consider
\begin{equation}
\label{eq:numerics_1d}
    \partial_t u + f'(u)\partial_x u = 0,
    \qquad x\in [0,L],
    \qquad u(x,0)=u_0(x).
\end{equation}
In two spatial dimensions, we consider the scalar conservation law
\begin{equation}
\label{eq:numerics_2d}
    \partial_t u
    +
    \partial_x f(u)
    +
    \partial_y f(u)
    =
    0,
    \qquad (x,y)\in [0,L]^2.
\end{equation}
The reference solution is computed by applying the matrix exponential to the semi-discrete level-set transport operator. The classical Schr\"odingerisation solution is computed by direct integration of the lifted Hamiltonian system. The Trotter--Schr\"odingerisation solution is obtained by the  circuit implementation described in the previous sections.

The main numerical parameters are
\[
N_x=2^{n_x},\qquad
N_p=2^{n_p},\qquad
N_w=2^{n_w},
\]
where \(N_x\) is the spatial grid size per dimension, \(N_p\) is the number of level-set grid points, and \(N_w\) is the number of auxiliary Fourier grid points introduced by Schr\"odingerisation. The parameter \(R\) denotes the truncation length in the auxiliary variable, and \(\Delta t\) denotes the product-formula step size in the Trotter implementation.

\subsection{1D traffic-flow example}
\label{subsec:numerics_traffic}

We first consider the one-dimensional traffic-flow-type flux
\[
    f(u)=u(1-u),
    \qquad
    f'(u)=1-2u.
\]
This flux is the classical Lighthill--Whitham--Richards flux when \(u\) is interpreted as a normalized density. In the present experiment, it is used as a nonlinear benchmark flux for the scalar conservation law \eqref{eq:numerics_1d}. The initial data is chosen as a smooth periodic profile,
\[
    u_0(x)=\sin(2\pi x/L),
\]
so that the solution remains smooth over the tested time interval and the accuracy of the Schr\"odingerisation recovery can be clearly observed.

Figure~\ref{fig:traffic_profile} compares the matrix-operator reference solution, the classical Schr\"odingerisation solution, and the Trotter--Schr\"odingerisation solution. The reference solution is shown as the baseline for the semi-discrete level-set equation. The classical Schr\"odingerisation curve tests the accuracy of the warped phase transformation and the recovery from the auxiliary variable. The Trotter--Schr\"odingerisation curve further includes the product-formula implementation error of the quantum circuit.

\begin{figure}[H]
    \centering
    \includegraphics[width=0.72\linewidth]{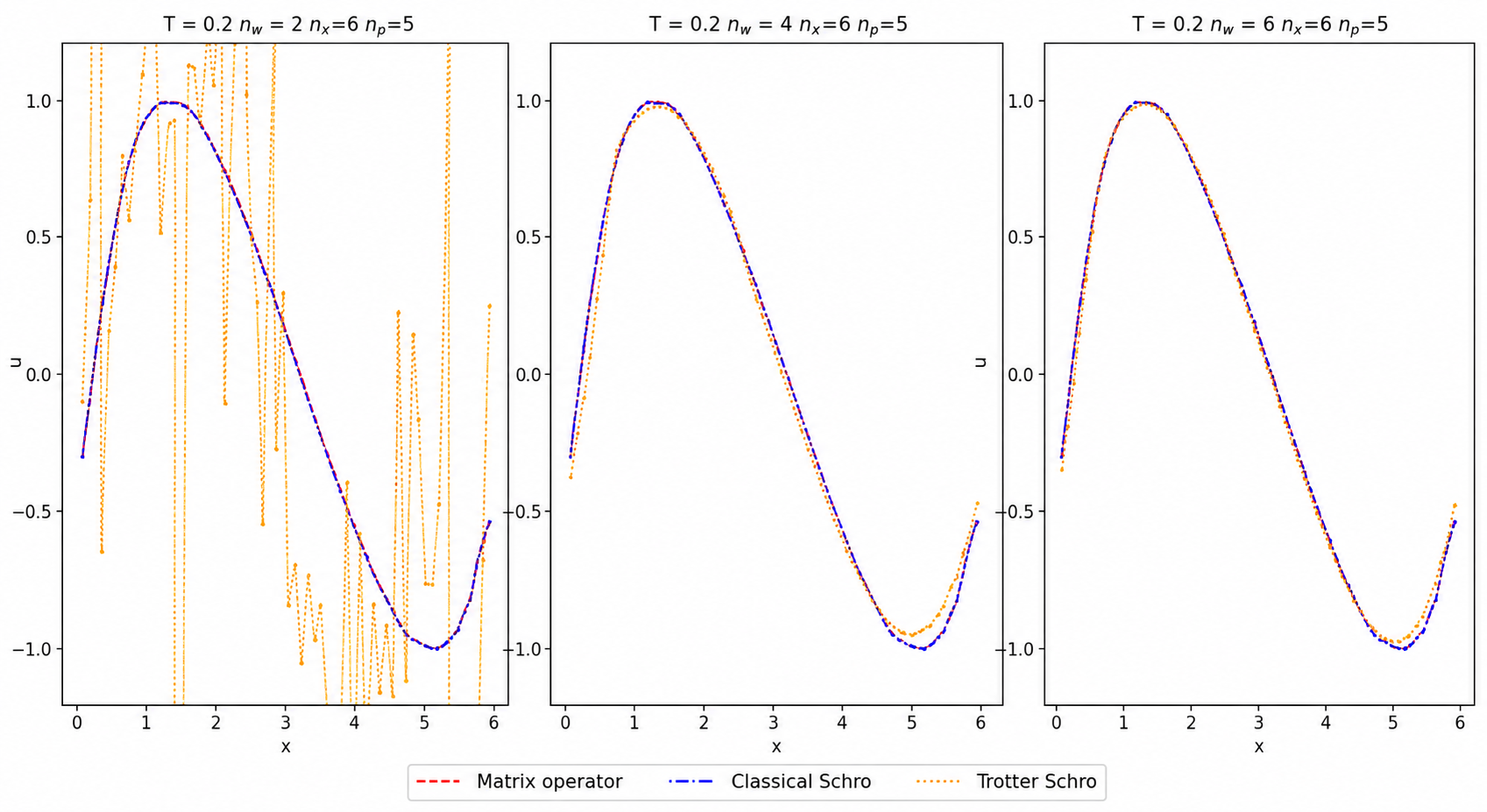}
    \caption{One-dimensional traffic-flow-type equation with \(f(u)=u(1-u)\). The figure compares the matrix-operator reference, the classical Schr\"odingerisation solution, and the Trotter--Schr\"odingerisation solution for different auxiliary resolutions \(n_w\).}
    \label{fig:traffic_profile}
\end{figure}

The results show that the auxiliary resolution \(n_w\) plays a decisive role in the recovery quality. For small \(n_w\), the auxiliary Fourier grid is too coarse to resolve the warped phase profile, and the Trotter--Schr\"odingerisation solution may exhibit visible oscillations. As \(n_w\) increases, the recovered profile becomes stable and agrees closely with the matrix-operator reference. The agreement between the classical Schr\"odingerisation and the Trotter--Schr\"odingerisation curves at sufficiently large \(n_w\) indicates that, at this resolution, the dominant error is no longer the auxiliary-variable recovery error, and the product-formula error is controlled for the chosen step size.

\subsection{2D Burgers equation}
\label{subsec:numerics_2d_burgers}

We next test the two-dimensional extension of the method using the Burgers flux
\[
    f(u)=\frac{u^2}{2}.
\]
Equation~\eqref{eq:numerics_2d} becomes
\begin{equation}
\label{eq:2d_burgers}
    \partial_t u
    +
    u\partial_x u
    +
    u\partial_y u
    =
    0.
\end{equation}
The initial data is a smooth localized Gaussian pulse,
\[
    u_0(x,y)
    =
    \exp\left(
        - (x-L/2)^2 - (y-L/2)^2
    \right).
\]
This example is useful for testing the multidimensional tensor-product structure of the algorithm. In the level-set representation, each \(p_l\)-slice satisfies
\[
    \partial_t \psi_l
    +
    p_l\partial_x\psi_l
    +
    p_l\partial_y\psi_l
    =
    0.
\]
Therefore, the semi-discrete transport operator takes the form
\[
    A_l
    =
    p_l(D_x+D_y),
\]
where \(D_x\) and \(D_y\) are the one-dimensional difference operators embedded into the two spatial directions.

Figure~\ref{fig:2d_burgers} reports the two-dimensional Trotter--Schr\"odingerisation solution for increasing values of \(n_w\). The upper row shows the recovered solution, while the lower row shows the pointwise absolute error with respect to the matrix-operator reference,
\[
    |u_{\rm Trotter}(T,x,y)-u_{\rm ref}(T,x,y)|.
\]

\begin{figure}[H]
    \centering
    \includegraphics[width=0.82\linewidth]{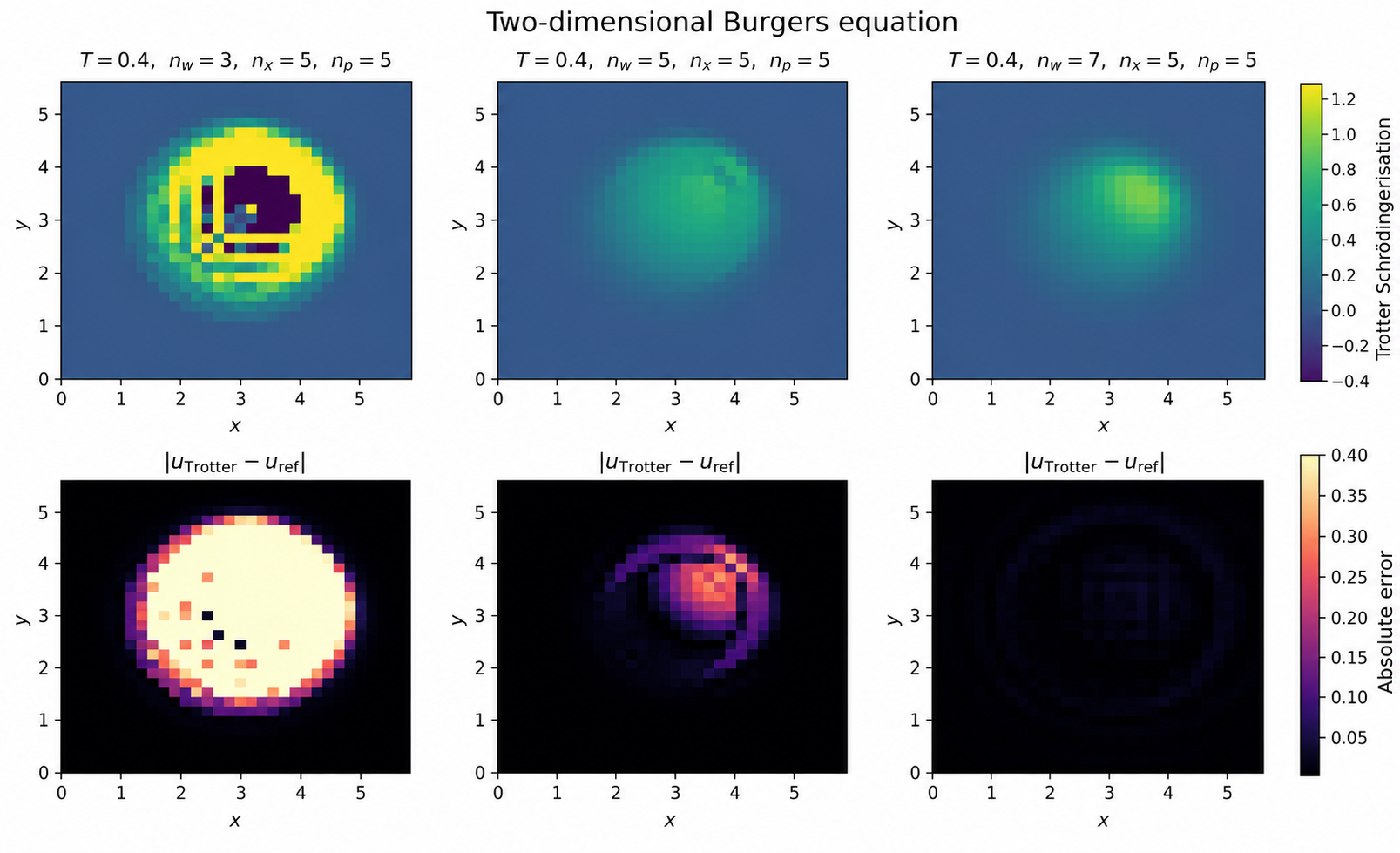}
    \caption{Two-dimensional Burgers equation with Gaussian initial data. The upper row shows the Trotter--Schr\"odingerisation solution for different values of \(n_w\). The lower row shows the absolute error with respect to the matrix-operator reference solution.}
    \label{fig:2d_burgers}
\end{figure}

The two-dimensional experiment is more demanding than the one-dimensional case for two reasons. First, the spatial operator contains contributions from both directions, \(D_x+D_y\), which increases the effective transport strength. Second, the auxiliary-variable recovery error is amplified in the two-dimensional lifted system because each \(p_l\)-slice now evolves on a larger spatial state space. As a result, small \(n_w\) may lead to visible oscillations or phase errors in the recovered solution.

Nevertheless, the qualitative behavior is consistent with the theoretical construction. Increasing \(n_w\) improves the resolution of the auxiliary Fourier variable and reduces the recovery error. The error plots show that the largest discrepancies are localized near the region where the Gaussian is transported and deformed. This is expected, since the level-set density is most concentrated near the moving front and the recovery formula is most sensitive in regions with large gradients. The experiment therefore confirms the multidimensional applicability of the formulation, while also illustrating that two-dimensional Trotter simulations require substantially higher auxiliary resolution than the corresponding one-dimensional tests.

\subsection{Gate complexity and convergence with respect to \(N_w\)}
\label{subsec:numerics_complexity_convergence}

We finally examine the computational scaling of the circuit implementation and the convergence with respect to the auxiliary Fourier resolution. Figure~\ref{fig:gate_scaling} shows the total gate count of the time-evolution circuit \(V_C(T)\) as a function of \(N_w\), for several spatial resolutions \(N_x\).

\begin{figure}[H]
    \centering
    \includegraphics[width=0.72\linewidth]{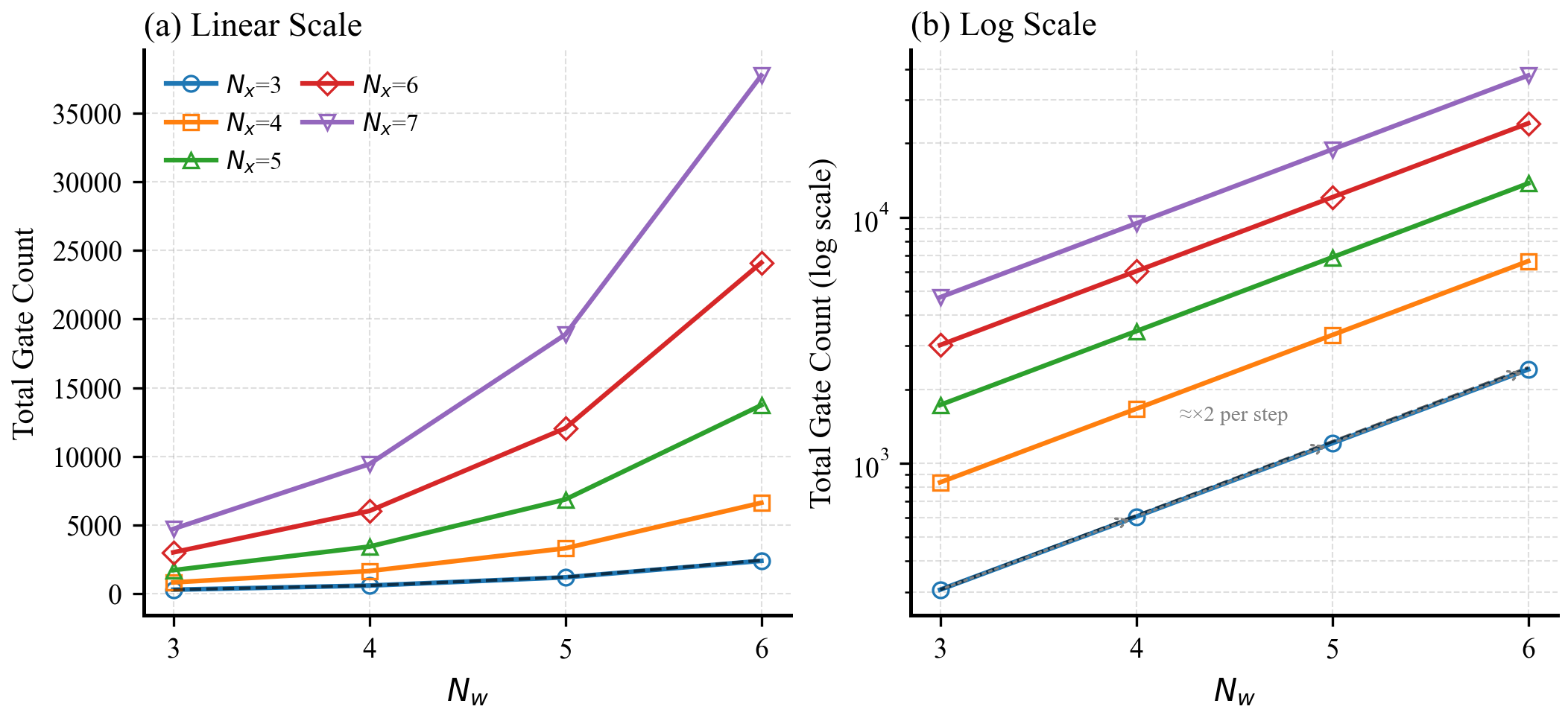}
    \caption{Total gate count of the time-evolution circuit as a function of the auxiliary Fourier grid size \(N_w\), for several spatial resolutions \(N_x\).}
    \label{fig:gate_scaling}
\end{figure}

Figure~\ref{fig:gate_scaling} provides a numerical verification of the gate-complexity estimate in Lemma~\ref{lem: gate complexity} using the one-dimensional Burgers equation as a representative
example. For fixed spatial resolution \(N_x\), increasing the auxiliary Fourier resolution \(N_w\) increases the number of controlled powers associated with the Schr\"odingerisation register.

Figure~\ref{fig:nw_convergence} shows the \(L_2\) error as a function of \(n_w\). The error is computed against the matrix-operator reference solution. This experiment isolates the effect of the auxiliary-variable discretization and recovery. The natural function
\[
    g(w)=e^{-|w|}
\]
is continuous but not differentiable at \(w=0\), and therefore gives only first-order convergence with respect to \(N_w\). Replacing it by a \(C^1\) cubic smoothing improves the observed convergence rate, in agreement with Remark~\ref{remarkofw}.

\begin{figure}[H]
     \centering
     \includegraphics[width=0.58\linewidth]{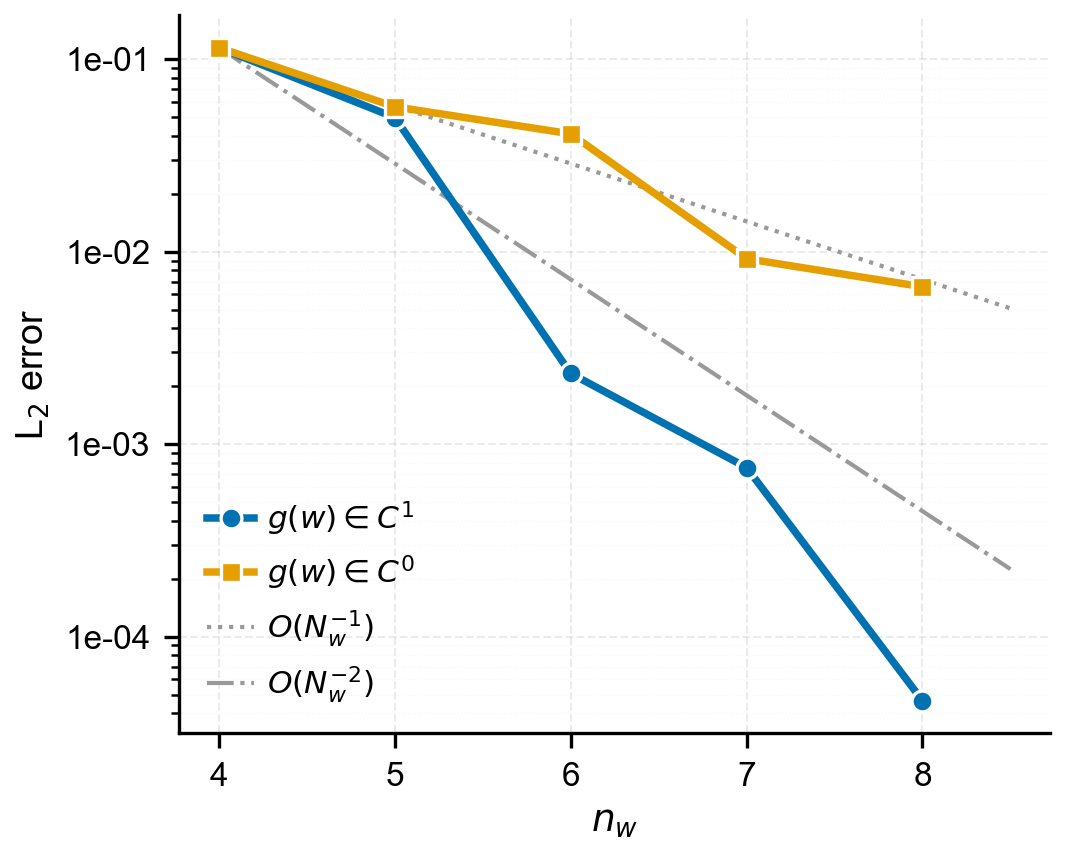}
     \caption{\(L_2\) error as a function of \(n_w\) for different choices of the auxiliary profile \(g(w)\). Smoother profiles improve the convergence rate in the auxiliary variable.}
     \label{fig:nw_convergence}
\end{figure}

The convergence behavior in Figure~\ref{fig:nw_convergence} supports the theoretical role of the auxiliary variable. A smoother choice of \(g(w)\) reduces the Fourier discretization error and improves the accuracy of the recovered solution. In principle, if \(g\in C^\infty\) and the solution remains sufficiently regular in the auxiliary variable, spectral convergence in \(N_w\) can be expected.

Together, Figures~\ref{fig:traffic_profile}--\ref{fig:nw_convergence} validate the main numerical ingredients of the method. The one-dimensional traffic example verifies the nonlinear level-set formulation and Trotter--Schr\"odingerisation recovery in a simple setting. The two-dimensional Burgers example demonstrates the tensor-product extension to higher spatial dimensions. The gate-count and convergence experiments confirm that the observed numerical behavior is consistent with the theoretical complexity and auxiliary-variable error analysis.

\section{Conclusion}
\hspace{1.5em} We developed a quantum algorithm for scalar conservation laws by combining the level-set formulation of \cite{Jin2022QuantumAF} with Schr\"odingerisation-based Hamiltonian simulation. The original $d$-dimensional nonlinear conservation law is first lifted to a $(d+1)$-dimensional linear Liouville equation. After finite-difference discretization, the resulting nonunitary linear system is transformed into Hamiltonian dynamics following the constructions in \cite{Hu2024QuantumCF} and \cite{Sato2024HamiltonianSF}. We then provided explicit quantum circuits for the time evolution and observable recovery, together with error and gate-complexity estimates. Numerical experiments for one- and two-dimensional examples validate the accuracy of the method, its multidimensional extension, and the predicted dependence on the auxiliary Fourier resolution. Future work will extend the present framework from scalar conservation laws to systems of conservation laws, reduce simulation costs through higher-order product formulas and optimized circuit designs, improve state-preparation and observable-recovery techniques, and develop resource-efficient implementations for fault-tolerant quantum architectures.

\bibliographystyle{IEEEtran}
\bibliography{citation}

\appendix
\section{The quantum circuits of $V_1(-\tau)$}\label{appV1circuits}
\hspace{1.5em}In this appendix, we provide the explicit construction of the quantum circuits for the operator $V_1(-\tau)$, which approximates the evolution $U_1(-\tau) = \exp(-i\mathbf{H}_1\tau)$. Based on the Lie-Trotter-Suzuki decomposition derived in Section \ref{section quantum circuits}, the total evolution is decomposed into three primary functional blocks: $U_{12}(\tau)$, $U_{11}(\tau)$, and the global phase shift $Ph(2\gamma_1\tau)$.\par
$V_1(-\tau)$:
\begin{equation*}
    \Qcircuit @C=1.5em @R=1.5em {
   & q_0 & &\qw & \multigate{3}{U_{12}(\tau)} & \qw &\multigate{3}{U_{11}(\tau)}&\qw &\multigate{3}{Ph(2\gamma_1\tau)}&\qw\\
   & q_1 & &\qw & \ghost{U_{12}(\tau)} & \qw &\ghost{U_{11}(\tau)}&\qw&\ghost{Ph(2\gamma_1\tau)}&\qw \\
   &\vdots& &\qw & \ghost{U_{12}(\tau)} & \qw &\ghost{U_{11}(\tau)}&\qw&\ghost{Ph(2\gamma_1\tau)}&\qw \\
   & q_{n_x-1} & &\qw & \ghost{U_{12}(\tau)} & \qw &\ghost{U_{11}(\tau)}&\qw&\ghost{Ph(2\gamma_1\tau)}&\qw\\
   }
\end{equation*}
~\\
\par The operator $U_{11}(\tau)$ is implemented using a sandwich structure involving basis transformation operators and a multi-controlled rotation gate. This corresponds to:$$U_{11}(\tau) = U_{n_x}(0) X_{n_x} CRZ^{1, \dots, n_x-1}(2\gamma_1\tau) X_{n_x} U_{n_x}(0)^\dagger$$
$U_{11}(\tau)$:
\begin{equation*}
    \Qcircuit @C=1.5em @R=1.5em {
   & q_0 & &\qw & \multigate{3}{U_{n_x}(0)} & \qw &\qw&\qw&\ctrl{1}&\qw&\qw&\qw &\multigate{3}{U_{n_x}(0)}&\qw\\
   & q_1 & &\qw & \ghost{U_{n_x}(0)} & \qw &\qw&\qw&\ctrl{1}&\qw&\qw&\qw&\ghost{U_{n_x}(0)}&\qw \\
   &\vdots& &\qw & \ghost{U_{n_x}(0)} & \qw &\qw&\qw&\ctrl{1}&\qw&\qw&\qw&\ghost{U_{n_x}(0)}&\qw \\
   & q_{n_x-1} & &\qw & \ghost{U_{n_x}(0)} & \qw &\gate{X}&\qw&\gate{RZ}&\qw&\gate{X}&\qw&\ghost{U_{n_x}(0)}&\qw\\
}
\end{equation*}
~\\
The operator $U_{12}(\tau)$ represents the product of local evolution terms:$$U_{12}(\tau) = \prod_{j=1}^{n_x} (I^{\otimes (n_x-j)} \otimes W_j(\gamma_1, \tau))$$where each component $W_j(\gamma_1, \tau)$ is defined as:$$W_j(\gamma_1, \tau) = U_j(0) CRZ^{1, \dots, j-1}(-2\gamma_1\tau) U_j(0)^\dagger$$This structure ensures that the spatial correlations defined by the discretized level-set equation are correctly mapped onto the quantum register through localized multi-controlled rotations.\par
$U_{12}(\tau)$:
\begin{equation*}
    \Qcircuit @C=1.5em @R=1.5em {
   & q_0 & &\qw & \gate{W_1} & \qw &\multigate{1}{W_2}&\qw& &\cdots& &\qw&\multigate{3}{W_{n_x-1}}&\qw\\
   & q_1 & &\qw & \qw  &\qw& \ghost{W_2}&\qw& &\cdots& &\qw&\ghost{W_{n_x-1}}&\qw \\
   &\vdots& &\qw & \qw  &\qw& \qw&\qw& &\cdots& &\qw&\ghost{W_{n_x-1}}&\qw  \\
   & q_{n_x-1}& &\qw & \qw  &\qw& \qw&\qw& &\cdots& &\qw&\ghost{W_{n_x-1}}&\qw \\
}
\end{equation*}
~\\
$W_j(\gamma_1,\tau)$:
\begin{equation*}
    \Qcircuit @C=1.5em @R=1.5em {
   & q_0 & &\qw & \multigate{4}{U_{j}(0)} & \qw&\ctrl{1}&\qw &\multigate{4}{U_{j}(0)}&\qw\\
   & q_1 & &\qw & \ghost{U_{j}(0)} & \qw &\ctrl{1}&\qw&\ghost{U_{j}(0)}&\qw \\
   &\vdots& &\qw & \ghost{U_{j}(0)} &\qw&\ctrl{1}&\qw&\ghost{U_{j}(0)}&\qw \\
   & q_{j-2} & &\qw & \ghost{U_{j}(0)} & \qw &\ctrl{1}&\qw&\ghost{U_{j}(0)}&\qw\\
   & q_{j-1} & &\qw & \ghost{U_{j}(0)} & \qw &\gate{RZ}&\qw&\ghost{U_{j}(0)}&\qw\\
}
\end{equation*}

$U_j(0)$:
\begin{equation*}
    \Qcircuit @C=1.5em @R=1.5em {
   & q_0 &  & \targ & \qw &\qw&\qw&\qw& \qw\\
   & q_1 &  &\qw&\targ &\qw&\qw&\qw&\qw\\
   &\vdots&  & \qw &\qw &\cdots& &\qw&\qw \\
   & q_{j-2}&  &\qw &\qw &\qw&\targ&\qw&\qw\\  
   & q_{j-1}& &\ctrl{-4}&\ctrl{-3}&\qw&\ctrl{-1}&\gate{H}&\qw
}
\end{equation*}
\section{The quantum circuits of $V_2(\tau)$}\label{appendix2}
\hspace{1.5em}In this appendix, we provide the explicit construction of the quantum circuits for the operator $V_2(\tau)$, which approximates the evolution $U_2(\tau) = \exp(i\mathbf{H}_2\tau)$. Based on the Lie-Trotter-Suzuki decomposition derived in Section \ref{section quantum circuits}, the total evolution is decomposed into three primary functional blocks: $U_{22}(\tau)$ and $U_{21}(\tau)$.\par
$V_2(\tau)$:
\begin{equation*}
    \Qcircuit @C=1.5em @R=1.5em {
   & q_0 & &\qw & \multigate{3}{U_{22}(\tau)} & \qw &\multigate{3}{U_{21}(\tau)}&\qw &\qw\\
   & q_1 & &\qw & \ghost{U_{22}(\tau)} & \qw &\ghost{U_{21}(\tau)}&\qw&\qw \\
   &\vdots& &\qw & \ghost{U_{22}(\tau)} & \qw &\ghost{U_{21}(\tau)}&\qw&\qw \\
   & q_{n_x-1} & &\qw & \ghost{U_{22}(\tau)} & \qw &\ghost{U_{21}(\tau)}&\qw&\qw\\
}
\end{equation*}
~\\
$U_{22}(\tau)$:
\begin{equation*}
    \Qcircuit @C=1.5em @R=1.5em {
   & q_0 & &\qw & \gate{W'_1(\gamma_2,\tau)} & \qw &\multigate{1}{W'_2(\gamma_2,\tau)}&\qw &\multigate{3}{W'_j(\gamma_2,\tau)}&\qw&\multigate{5}{W'_{n_x}(\gamma_2,\tau)}&\qw\\
   & q_1 & &\qw & \qw & \qw &\ghost{W'_2(\gamma_2,\tau)}&\qw&\ghost{W'_j(\gamma_2,\tau)}&\qw&\ghost{W'_{n_x}(\gamma_2,\tau)}&\qw \\
   &\vdots& &\qw & \qw & \qw &\qw&\qw&\ghost{W'_j(\gamma_2,\tau)}&\qw&\ghost{W'_{n_x}(\gamma_2,\tau)}&\qw \\
   & q_{j-1} & &\qw & \qw & \qw &\qw&\qw&\ghost{W'_j(\gamma_2,\tau)}&\qw&\ghost{W'_{n_x}(\gamma_2,\tau)}&\qw\\
   & \vdots & &\qw & \qw & \qw &\qw&\qw&\qw&\qw&\ghost{W'_{n_x}(\gamma_2,\tau)}&\qw\\
   & q_{n_x-1} & &\qw &\qw & \qw &\qw&\qw&\qw&\qw&\ghost{W'_{n_x}(\gamma_2,\tau)}&\qw\\
}
\end{equation*}
~\\
$W'_j(\gamma_2,\tau)$:
\begin{equation*}
    \Qcircuit @C=1.5em @R=1.5em {
   & q_0 & &\qw & \multigate{4}{U_j(-\frac{\pi}{2})} & \qw &\gate{X}&\qw &\ctrl{1}&\qw&\gate{X}&\qw&\multigate{4}{U_j(-\frac{\pi}{2})^\dagger}&\qw\\
   & q_1 & &\qw & \ghost{U_j(-\frac{\pi}{2})} & \qw &\gate{X}&\qw &\ctrl{1}&\qw&\gate{X}&\qw&\ghost{U_j(-\frac{\pi}{2})^\dagger}&\qw \\
   &\vdots& &\qw & \ghost{U_j(-\frac{\pi}{2})} & \qw &\gate{X}&\qw &\ctrl{1}&\qw&\gate{X}&\qw&\ghost{U_j(-\frac{\pi}{2})^\dagger}&\qw \\
   & q_{j-2} & &\qw & \ghost{U_j(-\frac{\pi}{2})} & \qw &\gate{X}&\qw &\ctrl{1}&\qw&\gate{X}&\qw&\ghost{U_j(-\frac{\pi}{2})^\dagger}&\qw\\
   & q_{j-1} & &\qw & \ghost{U_j(-\frac{\pi}{2})} & \qw &\qw&\qw &\gate{RZ}&\qw&\qw&\qw&\ghost{U_j(-\frac{\pi}{2})^\dagger}&\qw\\
}
\end{equation*}
~\\
$U_j(-\frac{\pi}{2})$:
\begin{equation*}
    \Qcircuit @C=1.5em @R=1.5em {
   & q_0 &  & \targ & \qw &\qw&\qw&\qw& \qw&\qw&\qw\\
   & q_1 &  &\qw&\targ &\qw&\qw&\qw&\qw&\qw&\qw\\
   &\vdots&  & \qw &\qw &\cdots& &\qw&\qw&\qw&\qw \\
   & q_{j-2}&  &\qw &\qw &\qw&\targ&\qw&\qw&\qw&\qw\\  
   & q_{j-1}& &\ctrl{-4}&\ctrl{-3}&\qw&\ctrl{-1}&\gate{P(-\frac{\pi}{2})}&\qw&\gate{H}&\qw
}
\end{equation*}

\section{Details of Remark \ref{rem: binary}}\label{appendix:remark}
Without loss of generality, after shifting and rescaling the \(p\)-grid, we take
\[
P_{\alpha,l}=l,\qquad l\geq 0,
\]

Writing
\[
l=\sum_{k=0}^{n_p-1}l_k2^k,
\]
the \(l\)-multiplexed gates can be replaced by binary-controlled powers.

First consider the \(V_2\)-part. 
we have
\begin{equation}\label{V2_linear_binary}
\sum_l |l\rangle\langle l|\otimes V_2^{(l)}(\tau)
=
\prod_{k=0}^{n_p-1}
\Lambda_{l_k}
\left[
\tilde{V}_2(\tau)^{2^k}
\right].
\end{equation}
Here $\tilde{V}_2(\tau)=\prod_{\alpha=1}^d
(V_2(\tau))_\alpha$ and \(\Lambda_{l_k}[U]\) denotes the controlled unitary
\[
\Lambda_{l_k}[U]
=
|0\rangle\langle0|_{l_k}\otimes I
+
|1\rangle\langle1|_{l_k}\otimes U .
\]
Thus the factor \(2^k\) in the binary expansion of \(l\) is implemented by a gate
controlled by the qubit \(l_k\).

The circuit for the \(V_2\)-part is
\[
\Qcircuit @C=1em @R=1.25em {
\lstick{q}
& \qw {/}^{dn_x}
& \gate{\tilde{V}_2^{2^0}(\tau)}
& \gate{\tilde{V}_2^{2^1}(\tau)}
& \qw
& \cdots
& 
& \gate{\tilde{V}_2^{2^{n_p-1}}(\tau)}
& \qw
\\
\lstick{l_0}
& \qw
& \ctrl{-1}
& \qw
& \qw
& \cdots
& 
& \qw
& \qw
\\
\lstick{l_1}
& \qw
& \qw
& \ctrl{-2}
& \qw
& \cdots
& 
& \qw
& \qw
\\
\lstick{\vdots}
& \qw
& \qw
& \qw
& \qw
& \ddots
& 
& \qw
& \qw
\\
\lstick{l_{n_p-1}}
& \qw
& \qw
& \qw
& \qw
& \cdots
& 
& \ctrl{-4}
& \qw
}.
\]

Next consider the \(w\)-controlled block
\[
\sum_l |l\rangle\langle l|\otimes V_w^{(l)}(-\tau),
\]\par 
As taking $ P_{\alpha,l}=l$ 
\begin{equation}
    \left(V_1^{(l)}(-\tau)\right)^{r-N_w/2}
=
\prod_{\alpha=1}^d
(V_1(l\tau(r-N_w/2)))_\alpha .
\end{equation}

Using the binary expansions $l=\sum_{k=0}^{n_p-1}l_k2^k$, $r=\sum_{m=0}^{n_w-1}r_m2^m$, we have
\[
l\left(r-\frac{N_w}{2}\right)
=
\left(\sum_{k=0}^{n_p-1}l_k2^k\right)
\left(
\sum_{m=0}^{n_w-1}r_m2^m-2^{n_w-1}
\right)=-\sum_{k=0}^{n_p-1}l_k2^{k+n_w-1}
+
\sum_{k=0}^{n_p-1}
\sum_{m=0}^{n_w-1}
l_kr_m2^{k+m}.
\]

Consequently,
\begin{equation}\label{Vw_linear_double_control}
\begin{aligned}
\sum_l |l\rangle\langle l|\otimes V_w^{(l)}(-\tau)
&=
\prod_{k=0}^{n_p-1}
\Lambda_{l_k}
\left[
\prod_{\alpha=1}^d
(V_1(-2^{k+n_w-1}\tau))_\alpha
\right]
\\
&\quad\times
\prod_{k=0}^{n_p-1}
\prod_{m=0}^{n_w-1}
\Lambda_{l_k,r_m}
\left[
\prod_{\alpha=1}^d
(V_1(2^{k+m}\tau))_\alpha
\right].
\end{aligned}
\end{equation}
Here
\[
\Lambda_{l_k,r_m}[U]
=
(I-|11\rangle\langle11|_{l_kr_m})\otimes I
+
|11\rangle\langle11|_{l_kr_m}\otimes U
\]
denotes the unitary \(U\) controlled jointly by the \(p\)-register qubit \(l_k\) and the \(w\)-register qubit \(r_m\).

The first product in \eqref{Vw_linear_double_control} implements the negative
offset
\[
-l\frac{N_w}{2}
=
-\sum_{k=0}^{n_p-1}l_k2^{k+n_w-1},
\]
while the second product implements the positive term
\[
lr
=
\sum_{k=0}^{n_p-1}
\sum_{m=0}^{n_w-1}
l_kr_m2^{k+m}.
\]

The corresponding circuit can be written compactly as
\[
\Qcircuit @C=1em @R=1.25em {
\lstick{q}
& \qw {/}^{dn_x}
& \gate{L_0}
& \gate{L_1}
& \qw
& \cdots
& 
& \gate{L_{n_p-1}}
& \gate{M_{0,0}}
& \gate{M_{0,1}}
& \qw
& \cdots
& 
& \gate{M_{n_p-1,n_w-1}}
& \qw
\\
\lstick{l_0}
& \qw
& \ctrl{-1}
& \qw
& \qw
& \cdots
& 
& \qw
& \ctrl{-1}
& \ctrl{-1}
& \qw
& \cdots
& 
& \qw
& \qw
\\
\lstick{l_1}
& \qw
& \qw
& \ctrl{-2}
& \qw
& \cdots
& 
& \qw
& \qw
& \qw
& \qw
& \cdots
& 
& \qw
& \qw
\\
\lstick{\vdots}
& \qw
& \qw
& \qw
& \qw
& \ddots
& 
& \qw
& \qw
& \qw
& \qw
& \cdots
& 
& \qw
& \qw
\\
\lstick{l_{n_p-1}}
& \qw
& \qw
& \qw
& \qw
& \cdots
& 
& \ctrl{-4}
& \qw
& \qw
& \qw
& \cdots
& 
& \ctrl{-4}
& \qw
\\
\lstick{r_0}
& \qw
& \qw
& \qw
& \qw
& \cdots
& 
& \qw
& \ctrl{-5}
& \qw
& \qw
& \cdots
& 
& \qw
& \qw
\\
\lstick{r_1}
& \qw
& \qw
& \qw
& \qw
& \cdots
& 
& \qw
& \qw
& \ctrl{-6}
& \qw
& \cdots
& 
& \qw
& \qw
\\
\lstick{\vdots}
& \qw
& \qw
& \qw
& \qw
& \ddots
& 
& \qw
& \qw
& \qw
& \qw
& \cdots
& 
& \qw
& \qw
\\
\lstick{r_{n_w-1}}
& \qw
& \qw
& \qw
& \qw
& \cdots
& 
& \qw
& \qw
& \qw
& \qw
& \cdots
& 
& \ctrl{-8}
& \qw
}.
\]
where
\[
L_k
=
\tilde{V}_1^{-2^{k+n_w-1}}(\tau),
\quad
M_{k,m}
=
\tilde{V}_1^{2^{k+m}}(\tau)
\]
Thus \(L_k\) is controlled by \(l_k\), while \(M_{k,m}\) is jointly controlled by
\(l_k\) and \(r_m\).

\section{Proof of Lemma\ref{lemma_epsilon_G}}\label{appendix3}

We split the error in \eqref{eq:epsilon_omega} into three parts:
\[
\begin{aligned}
\left|
\langle G\rangle(t,\bm{x}_{\bm j})
-
\langle G\rangle_{\bm j,h}(t)
\right|
&\leq
\underbrace{
\left|
\int G(p)\psi(t,\bm{x}_{\bm j},p)\,dp
-
\int G(p)\psi^\omega(t,\bm{x}_{\bm j},p)\,dp
\right|
}_{I}
\\
&\quad+
\underbrace{
\left|
\int G(p)\psi^\omega(t,\bm{x}_{\bm j},p)\,dp
-
h_p\sum_l G(p_l)\psi^\omega(t,\bm{x}_{\bm j},p_l)
\right|
}_{II}
\\
&\quad+
\underbrace{
\left|
h_p\sum_l G(p_l)
\left[
\psi^\omega(t,\bm{x}_{\bm j},p_l)
-
\psi_h^\omega(t,\bm{x}_{\bm j},p_l)
\right]
\right|
}_{III}.
\end{aligned}
\]
where \(\psi^\omega\) is the regularized distribution obtained by replacing $\delta(p-u)$ with $\delta_\omega(p-u)$.
\paragraph{Part I: mollification error.}
The normalized symmetric mollifier satisfies
\[
\int_{\mathbb R}z\delta_\omega(z)\,dz=0,
\qquad
\int_{\mathbb R}z^2|\delta_\omega(z)|\,dz
\leq C\omega^2.
\]
For a single-valued branch, Taylor's theorem and
$G\in W^{2,\infty}(\mathbb R)$ give
\begin{align}
I
&=
\left|
G(u)
-
\int_{\mathbb R}G(u+z)\delta_\omega(z)\,dz
\right| \notag\\
&\leq
\frac12\|G''\|_{L^\infty}
\int_{\mathbb R}z^2|\delta_\omega(z)|\,dz
\leq C_G\omega^2.
\label{eq:proof-mollification}
\end{align}
The same estimate holds branchwise for a smooth multivalued solution,
with the Jacobian weights absorbed into the constant.

\paragraph{Part II: quadrature error in the \(p\)-direction.}
Define
\[
f(p)=G(p)\psi^\omega(t,\bm{x}_{\bm j},p).
\]
Using the midpoint quadrature rule on the \(p\)-grid gives
\[
II
\leq
C h_p^2
\|f''\|_{L^1_p}.
\]
By the product rule,
\[
f''(p)
=
G''(p)\psi^\omega
+
2G'(p)\partial_p\psi^\omega
+
G(p)\partial_{pp}\psi^\omega.
\]
For mollifiers,
\[
\|\delta_\omega\|_{L^1}=O(1),
\qquad
\|\delta_\omega'\|_{L^1}=O(\omega^{-1}),
\qquad
\|\delta_\omega''\|_{L^1}=O(\omega^{-2}).
\]
Therefore,
\[
\|f''\|_{L^1_p}
\leq
C_G
\left(
1+\omega^{-1}+\omega^{-2}
\right)
\leq
C_G\omega^{-2}.
\]
Hence
\[
II
\leq
C_p\frac{h_p^2}{\omega^2}.
\]
\paragraph{Part III: spatial discretization error.}
Let \(\psi_h^\omega\) be the semi-discrete solution obtained by replacing
\(\partial_{x_\alpha}\) with first-order upwind difference operator
\(D^{(\alpha)}_+\) and \(D^{(\alpha)}_-\). The local truncation error satisfies
\[
D^{(\alpha)}_*\psi^\omega
-
\partial_{x_\alpha}\psi^\omega
=
O\left(
h\partial_{x_\alpha x_\alpha}\psi^\omega
\right).
\]
Therefore, for the transport operator
\[
F'(p)\cdot\nabla_{\bm x}
=
\sum_{\alpha=1}^d F_\alpha'(p)\partial_{x_\alpha},
\]
the local residual satisfies
\[
\|\tau_h(s)\|
\le
C d h
\max_{\alpha}
\|\partial_{x_\alpha x_\alpha}\psi^\omega(s)\|.
\]
By Duhamel's formula and stability of the semi-discrete upwind transport
semigroup,
\[
\|\psi_h^\omega(t)-\psi^\omega(t)\|
\le
\int_0^t
\|\tau_h(s)\|\,ds.
\]
Hence
\[
III(t)
\le
C d h
\int_0^t
\max_{\alpha}
\|\partial_{x_\alpha x_\alpha}\psi^\omega(s)\|\,ds.
\]
Using the regularized delta estimate gives
\[
\max_{0\le s\le t}
\max_{\alpha}
\|\partial_{x_\alpha x_\alpha}\psi^\omega(s)\|
\le
C_h\frac{1}{\omega^2},
\]
then
\[
III(t)
\le
C_x \frac{dht}{\omega^2}.
\]

Therefore,
\[
\varepsilon_\omega
\leq
C
\left(
\omega^2
+
\frac{h_p^2}{\omega^2}
+
\frac{d h t}{\omega^2}
\right).
\]
If \(h_p\sim h\), then \(h_p^2/\omega^2\) is lower order than
\(d h/\omega^2\) in the small mesh regime, and hence
\[
\varepsilon
\leq
C
\left(
\omega^2+\frac{d h t}{\omega^2}
\right).
\]
Balancing the two terms yields
\[
\omega\sim(dht)^{1/4},
\qquad
\varepsilon=O((dht)^{1/4}).
\]

\section{Proof of Theorem\ref{thm:totalcomplexity}}\label{appendix4}
The proof proceeds in three stages:
(I) determine the discretization parameters from the error budget,
(II) assemble the gate count, and
(III) substitute the mesh-size scaling into the final complexity estimate.

\paragraph{Stage I: Parameter determination from the error budget.}

By Corollary~\ref{cor: richardson_delta}, the Richardson-corrected delta
approximation gives
\[
\varepsilon_R
\leq
C_R\omega^4.
\]
Requiring \(\varepsilon_R\leq O(\varepsilon)\) gives
\[
\omega=O(\varepsilon^{1/4}).
\]

The same corollary gives the spatial discretization constraint
\[
\varepsilon_x
\leq
C_x\frac{d h T}{\omega^2}.
\]
Thus, requiring \(\varepsilon_x\leq O(\varepsilon)\) yields
\[
h
=
O\left(
\frac{\varepsilon\omega^2}{dT}
\right).
\]
Substituting \(\omega=O(\varepsilon^{1/4})\), we obtain
\[
h
=
O\left(
\frac{\varepsilon^{3/2}}{dT}
\right).
\]
Therefore,
\[
N_x=h^{-1}
=
O\left(
dT\varepsilon^{-3/2}
\right),
\qquad
n_x=\lceil\log_2 N_x\rceil
=
O(\log\varepsilon^{-1}),
\]
up to logarithmic factors in \(d\) and \(T\). Moreover,
\[
\beta_{\max}
=
\frac{\max_{\alpha,l}|F_\alpha'(p_l)|}{2h}
=
O(h^{-1}).
\]

For the \(w\)-discretization and truncation errors, the warped phase transformation uses the function $e^{-|w|}$ on $w \in [-R, R]$. The tail $w > R$ contributes error $\propto e^{-R}$. Requiring $e^{-R} \leq \varepsilon$ gives $R = \mathcal{O}(\log\varepsilon^{-1})$, which can be absorbed into $\widetilde{O}$. 

\paragraph{Stage II: Gate count assembly.}

By Lemma~\ref{lem: gate complexity}, the evolution circuit can be implemented
with
\[
Q_{\rm total}^{(1q)}
=
\widetilde O
\left(
d^2T^2N_pN_w^2h^{-2}\varepsilon^{-1}
\right),
\]
and
\[
Q_{\rm total}^{(2q)}
=
\widetilde O
\left(
d^2T^2N_pN_w^2h^{-2}\varepsilon^{-1}
\right).
\]
with
\[
\beta_{\max}=O(h^{-1}),
\]
The solution recovery circuit contributes only lower-order factors compared with the Hamiltonian evolution circuit and is absorbed into the \(\widetilde O\). 

\paragraph{Stage III: Substitution of the Richardson mesh constraint.}

From Stage I,
\[
h
=
O\left(
\frac{\varepsilon^{3/2}}{dT}
\right).
\]
Hence
\[
h^{-2}
=
O\left(
d^2T^2\varepsilon^{-3}
\right).
\]
Substituting this into the gate count gives
\[
Q_{\rm total}^{(1q)}
=
\widetilde O
\left(
d^2T^2N_pN_w^2
\cdot
d^2T^2\varepsilon^{-3}
\right),
\]
and therefore
\[
Q_{\rm total}^{(1q)}
=
\widetilde O
\left(
d^4T^4N_pN_w^2\varepsilon^{-4}
\right).
\]
The same argument gives
\[
Q_{\rm total}^{(2q)}
=
\widetilde O
\left(
d^4T^4N_pN_w^2\varepsilon^{-4}
\right).
\]
This completes the proof.

\end{document}